\documentclass[11pt]{article}

\usepackage[utf8]{inputenc} 
\usepackage[T1]{fontenc}    
\usepackage{times}
\usepackage{url}            
\usepackage{booktabs}       
\usepackage{amsfonts}       
\usepackage{nicefrac}       
\usepackage{microtype}      
\usepackage{bookmark}
\usepackage[margin=1in]{geometry}

\usepackage{algorithm}
\usepackage{algorithmic}

\usepackage{amsmath,amsthm}
\usepackage{bm}

\DeclareMathOperator*{\argmin}{argmin}

\newcommand{\X}{\mathbf{X}}

\newcommand{\Expect}[1]{\mathbb{E}\left[ #1\right]}
\newcommand{\tr}[1]{\textrm{tr}\left(#1\right)}

\usepackage{graphicx}

\renewcommand{\hat}{\widehat}
\newcommand{\norm}[1]{\left\lVert #1 \right\rVert}
\newcommand{\snorm}[1]{\lVert #1 \rVert}
\newtheorem{theorem}{Theorem}
\newtheorem{lemma}{Lemma}
\newtheorem{prop}{Proposition}
\newcommand{\R}{\mathbb{R}}

\usepackage[sort&compress,numbers]{natbib}
\usepackage{amsmath,amssymb}
\usepackage[dvipsnames]{xcolor}

\usepackage{xspace}
\newcommand{\methodname}{{SuffPCR}\xspace}
\newcommand{\superPC}{SPC\xspace}
\newcommand{\bstar}{\beta^*}
\renewcommand{\top}{\mathsf{T}}

\usepackage{mathtools}
\mathtoolsset{showonlyrefs}

\usepackage{hyperref}       
\hypersetup{colorlinks=true,citecolor=blue,linkcolor=blue,urlcolor=blue}
\usepackage{pifont}
\newcommand{\cmark}{\ding{52}}%
\newcommand{\xmark}{\ding{55}}%

\newcommand{\rev}[1]{\textcolor{Black}{#1}}

\begin{document}

\title{Sufficient principal component regression for pattern discovery in
  transcriptomic data} 
\author{
  Lei Ding\\
  Department of Statistics\\
  Indiana University
  \and
  Gabriel E. Zentner\\
  Department of Biology\\
  Melvin and Bren Simon Comprehensive Cancer Center\\
  Indiana University
  \and
  Daniel J. McDonald\\
  Department of Statistics\\
  University of British Columbia
 }

\maketitle

\begin{abstract}
  \noindent\textbf{Motivation:} 
  Methods for global measurement of transcript abundance such as microarrays and
  RNA-Seq generate datasets 
  in which the number of measured features far exceeds the number of
  observations. 
  Extracting biologically meaningful  
  and experimentally tractable insights from such data therefore requires
  high-dimensional prediction. Existing  
  sparse linear approaches to this challenge have been stunningly successful,
  but 
  some important issues remain.  
  These methods can fail to select the correct features, predict poorly relative
  to non-sparse alternatives,  
  or ignore any unknown grouping structures for the features.
  \vspace{6pt}
 
  \noindent\textbf{Results:} 
We propose a method called \methodname that yields improved
predictions in high-dimensional
tasks including regression and classification, especially 
in the typical context of omics with correlated features. 
\methodname first estimates sparse principal components and then estimates a linear model on the recovered subspace. 
Because the estimated subspace is sparse in the features, the resulting predictions will depend on only a small subset of genes.
\methodname works well on a variety of simulated and experimental transcriptomic
data, performing nearly optimally when the model assumptions are
satisfied. We also demonstrate near-optimal theoretical guarantees.\vspace{6pt}

\vspace{6pt}

\noindent\textbf{Availability:} 
Code and raw data are freely available at
\url{https://github.com/dajmcdon/suffpcr}.
\rev{Package documentation may be viewed at
  \url{https://dajmcdon.github.io/suffpcr}.}\vspace{6pt}

\noindent\textbf{Contact:} \href{daniel@stat.ubc.ca}{daniel@stat.ubc.ca}\vspace{6pt}
\end{abstract}

\section{Introduction}
\label{sec:intro}

Global transcriptome measurement with microarrays and RNA-Seq is a
staple approach in many areas of biological research and has yielded
numerous insights into gene regulation. Given data from such
experiments, it is often desirable to identify a small number of
transcripts whose expression levels are associated with a phenotype of
interest (for instance, disease-free survival of cancer
patients). Indeed, projects such as The Cancer Genome Atlas (TCGA)
have aimed to generate massive volumes of such data to enable
molecular characterization of various cancers. While these data are
readily available, their high-dimensional nature (tens of thousands of
transcript measurements from a single experiment) makes identification
of a compact gene expression signature statistically and
computationally challenging. While the identification of a minimal
gene expression signature is valuable in evaluating disease prognosis,
it is also helpful for guiding experimental exploration. In practical
terms, a set of five genes highly associated with a certain disease
phenotype can be characterized more rapidly, at lower cost, and in
more depth than a set of 50 or 100 such genes using genetic techniques
such as CRISPR knockout and cancer biological methods such as
xenotransplantation of genetically-modified cells into
mice. Therefore, this paper prioritizes selecting a small
subset of transcript measurements which still provide accurate
prediction of phenotypes. 

With these goals in mind, supervised linear regression techniques such as ridge
regression \citep{hoerl1970ridge}, the lasso \citep{Tibshirani1996},  elastic
net \citep{zou2005regularization},
or other penalized methods are often employed. More commonly,
especially in genomics applications, 
the outcomes of interest tend to be the result of groups of genes,
which perhaps together describe more complicated processes. Therefore, researchers often turn to
unsupervised methods such as principal component
analysis (PCA), principal component regression (PCR), and partial
least squares (PLS) for both preprocessing and as predictive models~\citep[e.g.,][]{cera2019genes,Harel297,kabir2017identifying,Traglia979}.

\rev{
%
In genomics, one may collect expression measurements for thousands of
genes from microarrays or RNA-Seq with the goal of predicting
phenotypes or class outcomes.
In these settings, the number of patients is much smaller than the number
of gene measurements, and researchers are interested in (1)
the accurate prediction of the phenotype, (2) the correct identification of a
handful of predictive genes, and (3) computational tractability.  
Among these properties, the correct identification of a small
number of predictive genes is of crucial importance in practice, since it can
lead biologists to further investigate specific genes through
CRISPR knockout or other techniques. 
It is this genetic pattern discovery for which our proposed methodology is 
intended: data with many more measurements than observations; the potential that
some of 
the measurements may be grouped or correlated; the existence
of either a continuous or discrete outcome we wish to predict; and the belief
that these predictions only depend on some small
collection of groups rather than the entire set of measurements.
}


\subsection{Recent related work}
\label{sec:review}

PCA has two main drawbacks when used in high dimensions. 
The first is that PCA is non-sparse, so it uses information from all
the available genes instead of selecting only those which are important, a
key objective in omics applications. That is, the right singular
vectors or ``eigengenes''~\citep{Alter2000} depend on all the genes
measured rather than a small collection. The second is that these
sample principal components are not consistent estimators of the  
population parameters in high dimensions
\citep{johnstone2009consistency}. This means essentially that when the
number of patients is smaller than the number of genes, even if the first
eigengene could perfectly explain the data, PCA will not be able to
recover it.

Modern approaches specifically for \rev{pattern discovery in the genomics
  context} such
as supervised gene shaving~\citep{HastieTibshirani2000}, tree
harvesting~\citep{HastieTibshirani2001}, and supervised principal
components
(\superPC)~\citep{bair2004semi,bair2006prediction,paul2008preconditioning}
seek to combine the presence of the phenotype with the structure
estimation properties of eigendecompositions on the gene expression measurements
using unsupervised techniques to obtain the best of both.  PLS is common in
genomics~\citep[e.g.]{chakraborty2019,LeekStorey2007},
though it remains uncommon
in statistics and machine learning, and its theoretical properties are
poorly understood.
Other recent PCA-based approaches for genetics, though not directly applicable
for prediction are SMSSVD~\citep{HenningssonFontes2018} and
ESPCA~\citep{MinLiu2018}.

\subsection{Contributions}

In this paper, we leverage the strong theoretical
properties associated with sparse PCA to improve predictive
accuracy for regression and classification problems in genomics. 
We avoid the strong assumptions necessary for \superPC, the current
state-of-the-art, while obtaining the benefits associated with sparse subspace
estimation.  
In the case that the phenotype is actually generated as a linear
function of a handful of genes,  our
method, \methodname, performs nearly optimally: it does as well as if
we had known which genes were relevant beforehand. 
Furthermore, we justify theoretically that our procedure can both predict
accurately and recover the correct genes.  
Our contributions can be succinctly summarized as follows:
\begin{enumerate}
    \item We present a methodology for discovering small sets of
      predictive genes using sparse PCA; 
    \item Our method improves the computational properties of existing
      sparse subspace estimation approaches to enable previously
      impossible inference when the number of genes is very large;
    \item We demonstrate state-of-the-art performance of our method in synthetic
      examples and with 
      standard cancer microarray measurements;
    \item We provide near-optimal theoretical guarantees.
\end{enumerate}
\rev{Our methodology can be used in a variety of genomic pattern discovery
  settings. One such example is a modified version of traditional differential
  expression analysis. If we have treatment and control measurements, the
  logistic version of our method is appropriate with the advantage that it
  examines the impact of one gene adjusted for the contributions of others.
  Additionally, with a continuous treatment, the detection power can be
  increased relative to using an artificial dichotomization.}
  
In \autoref{sec:background}, we motivate the desire for
\emph{sufficient} PCR relative to previous approaches and present details of 
\methodname. \autoref{sec:empirical} illustrates performance in simulated,
semi-simulated, 
and real examples and discusses the biological implications of our
methods for a selection of cancers. \autoref{sec:theory} gives
theoretically justifies our methods, providing guarantees for 
prediction accuracy and correct gene selection. \autoref{sec:con} concludes.

\paragraph{Notation.} 
We use bold uppercase letters to denote matrices, lowercase Arabic letters
to denote row vectors and scalars, and uppercase Arabic letters for
for random variables.
\rev{
  Let $Y$ be a random, real-valued $n$-vector of independent variables $Y_i$, and
$\mathbf{X}$ be the rowwise concatenation of i.i.d.\ draws $X_i$ from a
distribution on $\R^p$ with covariance $\boldsymbol{\Sigma}$.
We denote the observed realization of the outcome variable $Y$ as
  $y\in\R^n$. To be explicit in the genomics context, $\X$ is an $n\times p$
  matrix where each row is a set of
  transcriptomic measurements from RNA-Seq or microarrays for a patient while
  $y_i$ is an observed phenotype of interest for the $i^{\textrm{th}}$ patient.
Because $\X$ is a matrix, this symbol represents both a random
matrix and its realization. In the following, the meaning should be clear from
the context.
} 
We
assume, without loss of generality, that $\Expect{X_i}=0$,
and that the measurements $\mathbf{X}$ has been centred.
The singular value decomposition (SVD) of a matrix $\mathbf{A}$
is $\mathbf{A} = \mathbf{U}(\mathbf{A}) \boldsymbol{\Lambda}(\mathbf{A})
\mathbf{V}^\top(\mathbf{A})$.  
In the specific case of $\X$, we suppress the dependence on $\X$ in the notation
and write 
$\X = \mathbf{U} \boldsymbol{\Lambda} \mathbf{V}^\top$.
We write $\mathbf{A}_d$ to indicate the first $d$ columns of the
matrix $\mathbf{A}$ and $a_j$ to denote the $j^{th}$ row. In the
case of the identity matrix, we use a subscript to denote its
dimension when necessary: $\mathbf{I}_p$. Let
$\mbox{tr}(\mathbf{A})$ denote the sum of the diagonal entries of
$\mathbf{A}$ while $\snorm{\mathbf{A}}_F^2 = \sum_{ij} a_{ij}^2$ is the squared
Frobenius norm of $\mathbf{A}$. 
$\snorm{\mathbf{A}}_{2,0}$ denotes $(2,0)$-norm of $\mathbf{A}$,
that is the number of rows in $\mathbf{A}$ that have 
non-zero $\ell_2$ norm.
$\snorm{\mathbf{A}}_{1,1}$ is the sum of the row-wise
$\ell_1$-norms. 
\rev{Finally,  $\mathbf{1}(a)$ is the indicator function for the expression
$a$, taking value 1 if $a$ is true or 0 if not.}

\section{Motivation and methodology}
\label{sec:background}


Supervised Principal Components
\citep{bair2004semi,bair2006prediction,paul2008preconditioning} is widely used
for solving high-dimensional prediction and feature selection problems. 
It targets dimension reduction and sparsity simultaneously by first
screening genes (or individual mRNA probes) based on their marginal correlation with the
phenotype (or likelihood ratio test in the case of non-Gaussian
noise). Then, it performs PCA on this selected subset
and regresses the phenotype on the resulting components (possibly with additional penalization). 
This procedure is computationally simple, but, zero population marginal
correlation is neither necessary nor sufficient to guarantee that the
associated population regression coefficient is zero. To make this statement mathematically precise,
consider the linear model 
$
Y_i = X_i^\top \bstar + \epsilon_i,
$
where $Y_i$ is a real-valued scalar phenotype, $X_i$ is real-valued vector
of genes, $\bstar$ is the true (unknown) coefficient vector, and $\epsilon_i$
is a mean-zero error. Defining \rev{as above} $\textrm{Cov}(X_i,X_i) =
\boldsymbol{\Sigma}$, and $\textrm{Cov}(X_i,Y_i)=\Phi$, \rev{then, for this
procedure to correctly recover the true nonzero components of $\bstar$, it
requires}
\begin{equation}
  \label{eq:assumption}
\Phi_j=0 \Rightarrow
\bstar_{j}=(\boldsymbol{\Sigma}^{-1}\Phi)_j=0.
\end{equation}
In words, we assume that the dot product of the $j^{\textrm{th}}$ row of
the precision matrix with the marginal covariance between $x$ and $y$
is zero whenever the $j^{\textrm{th}}$ element of $\Phi$ is zero.
\rev{While reasonable
in some settings, this assumption frequently fails. For example, individual features may
only be predictive of the response in the presence of other features.}
To illustrate why this assumption fails for genomics problems, we
examine a motivating counterexample. Using mRNA measurements for acute
myeloid leukemia (AML, 
\citealt{bullinger2004gene}), we estimate both $\boldsymbol{\Sigma}^{-1}$ and
$\Phi$ and \rev{proceed as if these estimates are the true population quantities.}
To estimate $\Phi$, we use the empirical covariance and set all but the largest
$n=116$ values equal to zero, corresponding to an \rev{extremely sparse estimate.}
For $\boldsymbol{\Sigma}^{-1}$, we use the Graphical
Lasso~\citep{FriedmanHastie2008} for all $p=6283$ genes at different
sparsity levels ranging from 100\% sparse
($\hat{\boldsymbol{\Sigma}}^{-1}_{ij}=0$ for all $i\neq j$) to 95\%
sparse. We then create \rev{a pseudotrue}
$\bstar=\widehat{\boldsymbol{\Sigma}}^{-1}\widehat{\Phi}$ as 
in Equation~\eqref{eq:assumption}. This is essentially the most favorable
condition for SPC. \rev{To reiterate, in order to evaluate this assumption, we
  create $\bstar$ based on estimates from real genetics data that
  are highly sparse. But, as we will see below, because the inverse covariance
  matrix is not ``sparse in the right way'', SPC will have a very high false
  negative rate and ignore important genes.}

\begin{table}
  \centering
\begin{tabular}{@{}lrrrrrr@{}}
  \toprule
\% sparsity of $\widehat{\boldsymbol{\Sigma}}^{-1}$&  100 & 99.9 & 99.6 & 98.9 & 97.5 & 95.3 \\
  \midrule
\% non-zero $\bstar$'s& 1.8 & 3.3 & 8.4 & 23.5 & 50.2 & 77.9 \\ 
  False negative rate & 0.000 & 0.431 & 0.778 & 0.921 & 0.963 & 0.976 \\ 
\bottomrule
\end{tabular}
\caption{Illustration of the failure of Equation \eqref{eq:assumption} on the
AML data.}
  \label{tab:assumption}
\end{table}

\autoref{tab:assumption} shows the sparsity of $\hat{\mathbf{\Sigma}}^{-1}$, the
percent of non-zero regression coefficients, and the percent of non-zero
regression coefficients which are incorrectly ignored under the assumption (the
false negative rate).  
Even if the precision matrix is 99.9\%  sparse, the false
negative rate is over 40\%, meaning we find fewer than 60\% of
the true genes. If the sparsity of $\hat{\mathbf{\Sigma}}^{-1}$
is allowed to decrease only slightly, the false negative rate
increases to over 95\%. Clearly, this screening procedure will
ignore many important genes in even the most favorable conditions for
SPC.

More recent work has attempted to avoid this assumption.
\citet{ding2017predicting} uses the initially selected set of features to
approximate the information lost in the screening step via techniques from
numerical linear algebra. An alternative discussed in
\citet{piironen2018iterative} iterates the screening step with the prediction
step, adding back features which correlate with the residual. Finally,
\citet{tay2018principal} assumes that feature groupings are known and
and estimates separate subspaces for different
groups.
All these methodologies are tailored to perform well when $\Phi$ and $\bstar$
have particular compatible structures.  




On the other hand, it is important to observe that a sufficient
condition for $\bstar_j=0$ in Equation \eqref{eq:assumption} is that
the $j^{\textrm{th}}$ row of the left eigenvectors of $\boldsymbol\Sigma$ is
0. 
Based on this intuition, we develop sufficient PCR (abbreviated as \methodname)
which leverages this insight: row sparse eigenvectors 
imply sparse coefficients, and hence depend on only a subset of genes.
\methodname is tailored to the
case that $\X$ lies approximately on a low-dimensional linear manifold
which depends on a small subset of features. Because the linear
manifold depends on only some of the features, $\bstar$ does as well. 

\subsection{Prediction with principal components}
\label{sec:explainPC}

PCA is a canonical unsupervised
dimension reduction method when it is reasonable to imagine that $\X$ lies on (or near) a low-dimensional linear manifold. It finds the best $d$-dimensional approximation 
of the span of $\X$ such that the reconstruction error
in $\ell_2$ norm is minimized. This problem is equivalent to maximizing the  variance explained by the projection:
\begin{align}
\label{eq:opt}
  \max_{\mathbf{V}} &\quad\tr{\mathbf{S} \mathbf{V} \mathbf{V}^\top}&
 \mbox{subject to} &\quad  \mathbf{V}^\top \mathbf{V} = \mathbf{I}_d,
\end{align}
where $\mathbf{S} = \frac{1}{n}\X^\top \X$ is the sample covariance matrix. 
Let $\X = \mathbf{U} \boldsymbol{\Lambda} \mathbf{V}^\top$,
then the solution of this optimization problem is
$\mathbf{V}_d$, the first $d$ right singular vectors,
and the estimator of the first $d$ principal components is
$\mathbf{U}_d \boldsymbol{\Lambda}_d$ 
or $\X \mathbf{V}_d$ equivalently. 
Given an estimate of the principal components, principal component
regression (PCR) is simply ordinary least squares (OLS) regression of the phenotype on the derived components $\mathbf{U}_d \boldsymbol{\Lambda}_d$.
One can convert the lower-dimensional estimator, say $\hat\gamma$, back to the original space to reacquire an estimator of $\bstar$ as 
$\hat{\beta} = \mathbf{V}_d \hat{\gamma}_{d}$. Other generalized
linear models can be used place of OLS to find $\hat\gamma$.

\subsection{Sparse principal component analysis}
\label{sec:fps}

As discussed in \autoref{sec:review}, standard PCA works poorly in high dimensions. 
Much like the high-dimensional regression problem, estimating high-dimensional principal components is ill-posed without additional structure. 
To address this issue many authors have focused on different sparse PCA estimators for the case when $\mathbf{V}$ is sparse in some sense.
Many of these methods achieve this goal by adding a penalty to Equation \eqref{eq:opt}. 
Of particular utility for the case of PCR when $\bstar$ is sparse is to choose a penalty that results in row-sparse $\mathbf{V}$. This intuition is justified by the following result.
\begin{prop}
  Consider the linear model $Y_i = X_i^\top \bstar + \epsilon$ with $\mathrm{Cov}(X_i,X_i) = \boldsymbol{\Sigma}$.
  Let $\boldsymbol{\Sigma} =
  \mathbf{V}(\boldsymbol{\Sigma})\boldsymbol{\Lambda}(\boldsymbol{\Sigma})
  \mathbf{V}(\boldsymbol{\Sigma})^\top$ be the eigendecomposition of $\boldsymbol{\Sigma}$ with
  $\boldsymbol{\Lambda}(\boldsymbol{\Sigma})_{jj}=0$ for $j>d \in \mathbb{Z}^+$. Then
  $\snorm{v(\boldsymbol{\Sigma})_j}_2=0 \Rightarrow \bstar_j=0.$
  \label{prop:1}
\end{prop}
\noindent The proof is immediate. 
For any $j$, if $\snorm{v(\boldsymbol{\Sigma})_j}_2=0$, then every
element in $v(\boldsymbol{\Sigma})_j$ is 0, indicating the $j^{\textrm{th}}$
row of $\boldsymbol{\Sigma}^{-1}$ will be 0.
Since
$\bstar_{j}=(\boldsymbol{\Sigma}^{-1}\Phi)_j$ where
$\textrm{Cov}(X_i,y_i)=\Phi$, it also results in $\bstar_j=0$.
This result stands in stark contrast to the
assumption in Equation \eqref{eq:assumption}. This proposition gives a guarantee rather
than requiring an assumption: if the rows of $\mathbf{V}_d$ are sparse,
then $\beta_*$ is sparse. 
The same intuition can easily be extended to the case
$\boldsymbol{\Lambda}(\boldsymbol{\Sigma})_{jj}\geq 0$ for all $j$
given a gap between the $d^{\textrm{th}}$ and $(d+1)^{\textrm{st}}
$ eigenvalues.
In this setting, the natural analogue of PCA is the solution to:
\begin{align}
\label{eq:sparsePC}
    \max_{\mathbf{V}} &\quad \tr{\mathbf{S}\mathbf{V} \mathbf{V}^\top} - \lambda \norm{\mathbf{V}}^2_{2,0} &
    \mbox{subject to} &\quad  \mathbf{V}^\top \mathbf{V} = \mathbf{I}_d.
\end{align}
Solutions $\mathbf{\hat V}_d$ of Equation \eqref{eq:sparsePC} will give projection matrices onto the best $d$-dimensional linear manifold such that $\mathbf{\hat V}_d$ is row sparse.
However, this problem is NP-hard.

Many authors have developed different versions of sparse PCA. For
example, \citet{d2005direct} and \citet{zou2006sparse} focus on the first principal component and add additional principal components iteratively to account for the variation left unexplained by the previous principal components.
\citet{VuLei2013} derives a rate-minimax lower bound, illustrating that no estimator can approach the population quantity faster than, essentially, $q\sqrt{d/n}$ where $q$ is a deterministic function of $\boldsymbol{\Sigma}$.
Later work in \citet{vu2013fantope} proposes a convex relaxation to Equation \eqref{eq:sparsePC} which finds the first $d$ principal components simultaneously and nearly achieves the lower bound:
\begin{align}
\label{eq:fps}
    \max_{\mathbf{V}} &\quad\tr{\mathbf{S} \mathbf{V} \mathbf{V}^\top} - \lambda \snorm{\mathbf{V}\mathbf{V}^\top}_{1,1} &
    \text{subject to}  &\quad\mathbf{V}\mathbf{V}^\top \in \mathcal{F}^d,
\end{align}
where
$
    \mathcal{F}^d := \{\mathbf{V} \mathbf{V}^\top: \mathbf{0}  \preceq\mathbf{V} \mathbf{V}^\top \preceq \mathbf{I}_p \text{ and } \mbox{tr}(\mathbf{V} \mathbf{V}^\top)=d \}
$
is a convex body called Fantope, motivating the name Fantope Projection and
Selection (FPS). 
The authors solve the optimization problem in Equation \eqref{eq:fps} with an
alternating direction method 
of multipliers (ADMM) algorithm. 

For these reasons, FPS is known as the current state-of-the-art sparse PCA
estimator with the best performance. 
However, despite its theoretical justification, FPS is less useful in
practice for solving prediction tasks, especially 
in genomics applications 
with $p\gg n$ (rather than just $p>n$) for two reasons. 
First, the original ADMM algorithm has
per-iteration computational complexity $\mathcal{O}(p^3)$, which is a burden especially when $p$ is large. 
Secondly, because of the convex 
relaxation using Equation~\eqref{eq:fps} rather than Equation~\eqref{eq:sparsePC}, $\mathbf{\hat V}_d$ from FPS tends to be entry-wise sparse, but
infrequently row-wise sparse unless the signal-to-noise ratio (SNR) is very large
($q$ is a function of this ratio).
\rev{We give explicit formulas for the SNR under this model in the Supplement,
  but heuristically, the SNR captures how well the data is described by a
  $d$-dimensional subspace through the relative magnitude of
  $\tr{\mathbf{\Lambda}_d}$ compared to $p$.}
In genomics applications with low SNR, which is common, estimates $\hat\beta$ tend to have large numbers of non-zero
coefficients with very small estimated values.
Thus we design \methodname based on the insights from Proposition 1, utilizing the best sparse PCA estimator FPS, and further addressing both of these issues to achieve better empirical performance while maintaining theoretical justification. 

\subsection{Sufficient principal component regression}
\label{sec:method}

\begin{algorithm}[tb]
  \begin{algorithmic}[1]
    \STATE {\bfseries Input:} $\mathbf{X}$, $\mathbf{S}$, $y$, $d$, $\lambda$.
    \STATE $\mathbf{B} \leftarrow \mathbf{0}, \mathbf{C} \leftarrow \mathbf{0}$
    \COMMENT{Initialization}  
    \WHILE{not converged} 
    \STATE $\mathbf{A} \leftarrow \mbox{Proj}_{\mathcal{F}^d}\left(\mathbf{B} -
      \mathbf{C} + \mathbf{S}/\lambda\right)$ 
    \COMMENT{Approximate projection}
    \STATE $\mathbf{B} \leftarrow \mbox{Soft}(\mathbf{A} + \mathbf{C})$
    \COMMENT{Elementwise soft-thresholding} 
    \STATE $\mathbf{C} \leftarrow \mathbf{C} + \mathbf{A} - \mathbf{B}$
    \ENDWHILE
    \STATE Decompose $\mathbf{B} =
    \mathbf{V}_d\boldsymbol{\Lambda}_d\mathbf{V}^\top_d$
    \COMMENT{Rank $d$ eigendecomposition}
    \STATE Compute $l = \textrm{diag}(\mathbf{V}_d \mathbf{V}^\top_d)$, sort in
    descending order 
    \STATE Choose $t$ by applying \autoref{alg:find_t} to $l$
    \STATE Set rows in $\mathbf{V}_d$ whose $\ell_2$ norm is smaller
    than $t$ as 0, and get $\hat{\mathbf{V}}_d$
    \STATE Solve $\hat\gamma = \argmin_\gamma
    \lVert y-\mathbf{X}\hat{\mathbf{V}}_d\gamma \rVert_2^2$
    \STATE {\bf Return:} $\hat\beta = \hat{\mathbf{V}}_d\hat\gamma$
     \end{algorithmic}
  \caption{\methodname (regression version)}
  \label{alg:suffpcr}
\end{algorithm}

\begin{figure}
  \centering
  \includegraphics[width=.95\linewidth]{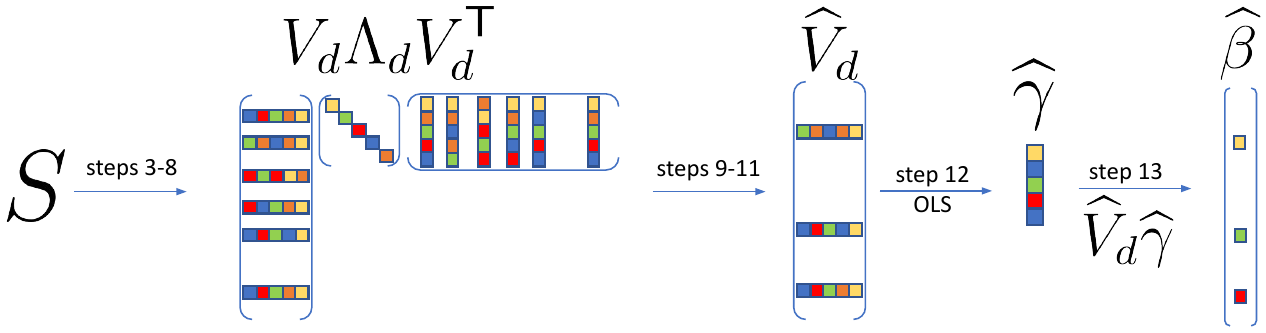}
  \caption{Graphical depiction of \autoref{alg:suffpcr}. Solid colours
    represent nonzero matrix entries.}
  \label{fig:algo1}
\end{figure}

  

In this section, we introduce \methodname.
The main idea of \methodname is to detect the relationship between the phenotype
$Y$ and gene expression measurements $\X$ by making use of the (near)
low-dimensional manifold that supports $\X$. 
In broad outline, \methodname first uses a tailored version of FPS to produce a
row-sparse estimate $\hat{\mathbf{V}}_d$, and then regresses $Y$ on the derived
components to produce sparse coefficient estimates.  
\methodname for regression is stated in \autoref{alg:suffpcr} and
summarized visually in \autoref{fig:algo1}.
For ease of exposition, we remind the reader that $Y$ and $\X$ are standardized so that $\mathbf{S} = \frac{1}{n} \X^\top \X$ is the correlation matrix. 
%
%

The first issue is the time complexity of the original FPS algorithm.
Essentially, FPS uses the same three steps depicted in lines 4, 5, and 6 in
\autoref{alg:suffpcr}. 
\begin{enumerate}
\item[$4'$.] $\mathbf{A} \leftarrow \mbox{Proj}_{\mathcal{F}^d}\left(\mathbf{B} - \mathbf{C} + \mathbf{S}/\lambda\right)$
\item[$5$.] $\mathbf{B} \leftarrow \mbox{Soft}(\mathbf{A} +
  \mathbf{C})$ where $\mbox{Soft}(b) = \mbox{sign}(b)\max\{|b|-1,\ 0\}$
\item[$6$.] $\mathbf{C} \leftarrow \mathbf{C} + \mathbf{A} - \mathbf{B}$.
\end{enumerate}
The only difference here between our implementation and that in FPS is
in step 4.
Each of these steps takes a matrix and produces another
matrix, where the matrices have $p^2$ elements. The second and third steps are
computationally simple 
(element-wise soft-thresholding and matrix addition). But the first,
$\mbox{Proj}_{\mathcal{F}^d}(\mathbf{Q})$,  
is more challenging. The solution requires computing the
eigendecomposition of $\mathbf{Q}$, an
$\mathcal{O}(p^3)$ operation, and then modifying the 
eigenvalues of $\mathbf{Q}$ through the solution of a
piecewise linear equation in $\tau$:
$
\mathbf{\Lambda}^2_{i,+}(\mathbf{Q})=
\min\{\max\{\mathbf{\Lambda}^2_{i}(\mathbf{Q})-\tau,\ 0\},\ 1\},
$
with $\tau$ such that  $\sum_{i=1}^{\min\{n,p\}} \mathbf{\Lambda}^2_{i,+}(\mathbf{Q})
= d$. The final result is then reconstructed as $\mathbf{A} = \mathbf{U}(\mathbf{Q})\mathbf{\Lambda}^2_{+}(\mathbf{Q})\mathbf{U}(\mathbf{Q})^\top$.
Because of the cubic complexity in $p$, the authors suggest the number of
features should not exceed one thousand. But typical transcriptomics data has many thousands of gene measurements, and preliminary selection of a subset is
suboptimal, as illustrated above.
Due to the form of the piecewise
solution, most eigenvalues will be set to 0. Thus, while we will
generally require more than $d$ eigenpairs, most are unnecessary,
certainly fewer than $\min\{n,p\}$.
Our implementation computes only a handful of eigenvectors corresponding to the
largest eigenvalues, rather than all $p$. If we compute enough to ensure that
some $\mathbf{\Lambda}^2_{i,+}(\mathbf{Q})$ will be 0, then
the rest are as well. Our implementation uses Augmented Implicitly
Restarted Lanczos Bidiagonalization~\citep[AIRLB,][]{BaglamaReichel2005} as
implemented in the \texttt{irlba} package~\citep{irlba-package}, though
alternative techniques such as those in
\citet{HomrighausenMcDonald2016,Gittens:2013aa} may
work as well. We provide a more detailed discussion
in the Supplement.

For moderate problems ($n,p \approx 100$), the
truncated eigendecomposition with AIRLB rather than
the full eigendecomposition leads to a three-fold speedup while the further
incorporation of specialized initializations leads to an eight-fold improvement
without any discernable loss of accuracy (results on a 2018 MacBook
Pro with 2.7 GHz Quad-Core processor and 16GB of memory running maxOS 10.15).
The results are similar when 
$p=5000$, though the same experiment on a high-performance Intel Xeon
E5-2680 v3 CPU with 12 cores, 256GB of memory, and optimized BLAS were
somewhat less dramatic (improvements of three-fold and four-fold
respectively). For large RNA-Seq datasets ($p\approx 20000$), we
observed a nearly ten-fold improvement in computation time.

The second issue is that the Fantope constraint in Equation
\eqref{eq:fps} ensures only that $\mbox{tr}(\mathbf{V}\mathbf{V}^\top)
= d$ but not that the number of rows with non-zero $l_2$-norm is
small. This feature of the convex relaxation results in
many rows with small, but non-zero, row-norm resulting in dense
estimates of $\bstar$. 
Thus, to make the final estimator $\hat{\mathbf{V}}_d$ sparse,
we hard-threshold rows in $\hat{\mathbf{V}}_d$ whose
$\ell_2$ norm is small, as illustrated in line 9, 10, and 11 in \autoref{alg:suffpcr}.
From empirical experience, we have found that there is often a strong elbow-type
behavior in the row-wise $\ell_2$ norm of $\hat{\mathbf{V}}_d$, similar to the
Skree plot used to choose $d$ in standard PCA.
Therefore, we develop a simple procedure, \autoref{alg:find_t}, to find the
best threshold automatically. Essentially, it calculates the empirical derivative of the
observation-weighted variances on each side of a potential threshold and  
maximizes their difference, resulting in signal and noise groups. We set the
rows in $\widehat{\mathbf{V}}_d$ corresponding 
to the noise  to 0. 
\methodname is also amenable for solving other generalized linear models.
For example, replacing line 12 in \autoref{alg:suffpcr} with logistic regression
solves classification problems.

\begin{algorithm}[tb!]
  \color{Black}
  \begin{algorithmic}[1]
    \STATE {\bfseries Input:} a $p$-vector $l$
    \FOR{$i \in 1, \cdots, p$}
    \STATE $T_n[i] = \texttt{var}(l[1:i])$
    \STATE $T_s[i] = \texttt{var}(l[(i+1):p])$ 
    \STATE $T[i] = i * T_n[i] + (p-i) T_s[i]$ 
    \STATE $\delta[i] = T[i] - T[i-1]$ \COMMENT{empirical derivative of $T$}
    \ENDFOR
    \STATE Set $i^*=\argmin_i  \{\delta[i] - \delta[i-1] >
    \texttt{mean}(|\delta[1:(i-1)]|)\}$ 
    \STATE {\bf Return:} $t=l[i^*]$
  \end{algorithmic}
  \caption{Find a $t$ to hard-threshold $l$}
  \label{alg:find_t}
\end{algorithm}

\section{Empirical evaluations}
\label{sec:empirical}

In this section, we show how \methodname performs on synthetic data
and on real public genomics datasets relative to state-of-the-art methods. 
\autoref{sec:sim} first presents a generative model for synthetic data and motivates the assumptions required for our theoretical results in \autoref{sec:theory}.
We include here one synthetic experiment under conditions favorable to \methodname
relative to \superPC.
We also 
investigate conditions favorable to \superPC, the influence of tuning parameter
selection, and the effect of the signal to
noise ratio but defer these to the Supplement.
\autoref{sec:semi} uses the NSCLC data as the $\mathbf{X}$ matrix, but creates
the response from a linear model.
\autoref{sec:real} reports the performance of \methodname on 5 public genomics
datasets.  
The supplement includes similar results for binary survival-status
outcomes. Across most settings in both synthetic and real data,
\methodname outperforms all competitors in prediction mean-squared error and is
able to select the true genes (those with $\bstar\neq 0$) more accurately.
\rev{An R package implementing \methodname and raw data are freely available at
\url{https://github.com/dajmcdon/suffpcr}. Package documentation may be viewed
at \url{https://dajmcdon.github.io/suffpcr}.}

\begin{figure*}[tb]
    \centering
    \includegraphics[width=0.95\linewidth]{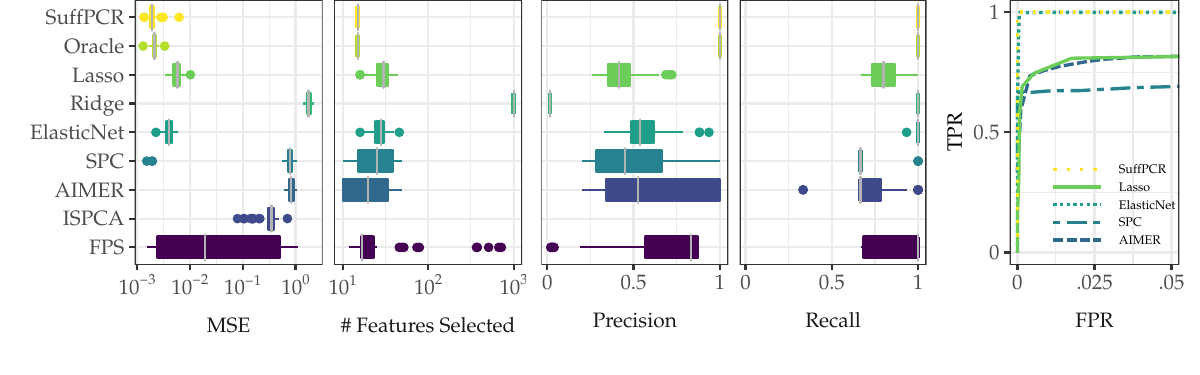}
    \caption{This figure compares the performance of \methodname against
      alternatives when the features come from a row-sparse factor model under
      favorable conditions for \methodname. Boxplots and receiver operating
      characteristic (ROC) curve (far right
      figure) are over 50 replications.  We have
      omitted the other methods from the ROC curve for legibility, but their
      behavior is similar to lasso. \rev{TPR and FPR stand for True/False Positive
      Rate respectively. Note that (as one would expect from the simulation
      conditions), SPC has the worst performance in terms of the ROC curve
      while both \methodname and Elastic net have AUC of almost 1.}
    } 
    \label{fig:reg1}
\end{figure*}



\begin{figure*}[tb]
    \centering
    \includegraphics[width=0.95\linewidth]{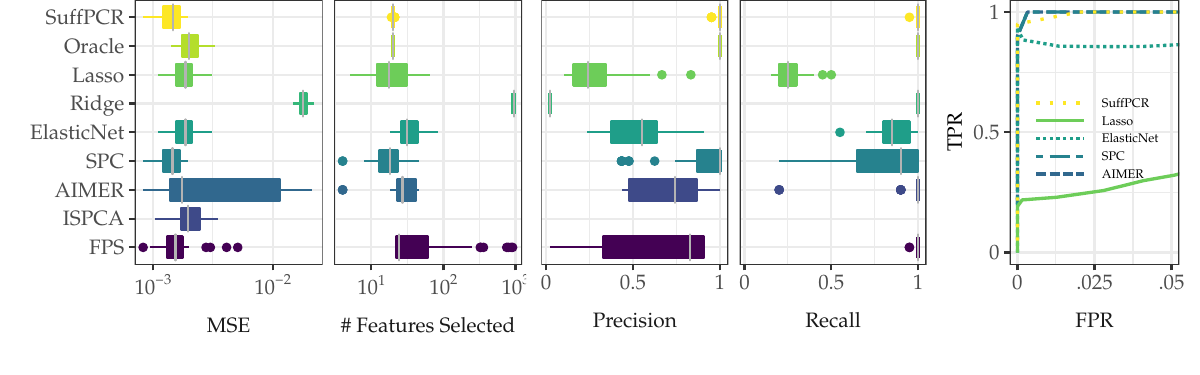}
    \caption{This figure compares the performance of \methodname against
      alternatives when the features come from a row-sparse factor model
      extracted from the NSCLC data. Boxplots and ROC curve (far right figure)
      are over 50 replications. \rev{In terms of the ROC curve, SPC and AIMER
        has the best 
      performance, though \methodname is not far behind. But note that SPC has
      much worse precision and recall.} } 
    \label{fig:semi1}
\end{figure*}

\subsection{Synthetic data experiments}
\label{sec:sim}

We generate data from the multivatiate Gaussian linear model
$
    y_i = x_i^\top \bstar + \epsilon_i,
$ 
where $x_i \sim \mbox{N}_p (0, \mathbf{\Sigma} )$,
$\bstar$ is the $p$-dimensional regression coefficient, $\epsilon_i \sim \mbox{N}(0,\sigma_y^2)$. 
We impose an orthogonal factor model for the covariates 
%
\rev{$    x_i = u^\top_i \mathbf{\Lambda}_d \mathbf{V}_d^\top + e_i,  $
where  $u_i$ are generated from $ \mbox{N}_{d}(0,
\mathbf{I}_d)$ independently, $\mathbf{\Lambda}_d$ is a diagonal
matrix with entries $(\lambda_1, \ldots, \lambda_d)$ in descending
order, and $\mathbf{V}_d\in\R^{p\times d}$ with $\mathbf{V}_d^\top
  \mathbf{V}_d = \mathbf{I}_d$.
The vector $e_i\in\R^n$ has i.i.d.\ $\mbox{N}(0, \sigma_x^2)$
entries independent of $u_i$, and $\sigma_x > 0$.}
We assume $\mathbf{V}_d$ is row sparse with only $s$ rows containing
non-zero entries. These non-zero rows are the ``true'' features to be
discovered, and they correspond to $\bstar \neq 0$.

\rev{It is important to note that, under this model, the rows of $\X$ follow a
  multivariate Gaussian 
  distribution independently, with mean 0 and full-rank covariance $\mathbf{\Sigma}
  = \mathbf{V}  \mathbf{L} \mathbf{V}^\top $ whenever $\sigma^2_x > 0$. 
  Here
  the columns of $\mathbf{V}$ are orthonormal
  eigenvectors on $\mathbb{R}^p$ and the
  eigenvalues are $l_1 \ge \cdots \ge l_p \ge 0$.
  Straightforward calculation shows that the first $d$ columns in $\mathbf{V}$ 
  are the same as the right singular vectors $\mathbf{V}_d$ in the signal
  component of $\mathbf{X}$.
  Furthermore, $l_i =\lambda_i^2 \mathbf{1}(i\leq d) + \sigma^2_x$, $i=1,\ldots,p$.}

We generate $y\in\R^n$ as a linear function of the latent factors
$\mathbf{U}_d$ with additive Gaussian noise:  
$
    y = \mathbf{U}_d   \Theta + z,
$
where $\Theta$ is the regression coefficient, and $z_i$ are i.i.d.\ $\mbox{N}(0,
\sigma^2_y)$, $i=1,\ldots,n$, independent of $\X$. 
%
%
Under this model
the population marginal correlation between each feature in $\X$ and $y$ is
$
  \Phi 
        = 
        \mathbf{V}_d \mathbf{\Lambda}_d \Theta,
$
and the population ordinary least squares coefficient of 
regressing $y$ on $\X$ is 
$
    \bstar = \mathbf{V}_d \mathbf{L}_d^{-1} \mathbf{\Lambda}_d \Theta.  
$
Note that the number of non-zero $\beta^*$ is $s$,
because $\mathbf{V}_d$ has only $s$ rows with non-zero entries.

In all cases, we use $n = 100$ observations and $p = 1000$ features, generating
three equal-sized sets for training, validation, and testing. We use
prediction accuracy on the validation set
to select \rev{tuning parameters for all methods. For the case of \methodname,
this means only $\lambda$, because we choose $t$ with \autoref{alg:find_t} and
set $d=3$. 
We use the test set for evaluating out-of-sample
performance.}
Each simulation is repeated 50 times. Results with $n=200$ and
$p=5000$ were similar. \rev{\autoref{alg:synthetic-data} makes this entire
  procedure more explicit.}


We compare \methodname with a number of alternative methods. The Oracle
estimator uses ordinary least squares (OLS) on the true features and
serves as a natural baseline: it uses information unavailable to the
analyst (the true genes) but represents the best method were that
information available. We also present results for Lasso \citep{Tibshirani1996}, Ridge \citep{hoerl1970ridge},
Elastic Net \citep{zou2005regularization}, \superPC
\citep{bair2006prediction}, AIMER \citep{ding2017predicting}, ISPCA \citep{piironen2018iterative}, and PCR using FPS directly without feature screening (using \autoref{alg:suffpcr} without step 9, 10 and 11). 
For ISPCA, we use the \texttt{dimreduce} R package to estimate the
principal components before performing regression.
\rev{For all competitors, we choose any tuning parameters that do not have
  default values using the validation set. Examples are $\lambda$ in Lasso,
  Ridge, and Elastic Net or the initial thresholding step in \superPC. We use
  the correct embedding dimension ($d=3$) whenever this is meaningful.} Additional
experiments are given in the Supplement. \rev{There, we investigate conditions
  favourable to \superPC, the choice of $d$, and the impact of different SNR choices.}

\begin{algorithm}[tb!]
  \color{Black}
  \begin{algorithmic}[1]
    \STATE {\bfseries Input:} $n=100$, $p=1000$, $r=5$, $d=3$,
    $\mathrm{SNR}_x=\mathrm{SNR}_y=5$.
    \STATE Generate i.i.d.\ $\textrm{N}(0,1)$ $\mathbf{U}\in\R^{n\times d}$,
    $\mathbf{E}\in\R^{n\times p}$, $z\in\R^n$.
    \STATE Set $\mathbf{\Lambda}_d = \textrm{diag}((d,\ d-1,\ldots,1)) \in \R^{d\times
      d}$.
    \STATE Generate i.i.d.\ $\textrm{N}(0,1)$ $\widetilde{\mathbf{V}}\in\R^{d\times
      d}$ and orthogonalize the columns.
    \STATE Extend $\widetilde{\mathbf{V}}\in\R^{s\times d}$ by repeating each
    row $r$ times ($s=rd$).
    \STATE Set $\mathbf{V}^\top_d = [\widetilde{\mathbf{V}}^\top\; \mathbf{0}] \in
    \R^{d \times p}$.
    \STATE Generate i.i.d.\ $\textrm{N}(0,1)$ $\tilde\Theta\in\R^{d-1}$.
    \STATE Set $\Theta_d = -
    \left(\sum_{i=1}^{d-1}\widetilde{\mathbf{V}}^\top_{ri}
      \mathbf{\Lambda}_{ii}\tilde{\Theta}_i\right) /
    (\widetilde{\mathbf{V}}_{rd} \mathbf{\Lambda}_{dd})$.
    \STATE Set $\Theta = [\tilde{\Theta}^\top\; \Theta_d]^\top$.
    \STATE Set $\beta^* = \mathbf{V}_d \mathbf{L}_d^{-1} \mathbf{\Lambda}_d \Theta$.
    \STATE Set $\sigma^2_x = \textrm{tr}(\mathbf{\Lambda}_d^2) /
    (p\mathrm{SNR}^2_x)$
    \STATE Set $\sigma^2_y =
    \left(\beta^{*\top}\mathbf{V}_d^\top\mathbf{\Lambda}_d^2 
      \mathbf{V}_d\bstar + \sigma^2_x\snorm{\bstar}_2^2\right) / (n\mathrm{SNR}^2_y)$.
    \STATE Set $\X =  \mathbf{U}_d \mathbf{\Lambda}_d \mathbf{V}_d^\top +
    \sigma_x\mathbf{E}$ and $y = \mathbf{U}_d   \Theta + \sigma_y z$
  \end{algorithmic}
  \caption{Generate synthetic data}
  \label{alg:synthetic-data}
\end{algorithm}

\subsubsection{Conditions favorable to \methodname}
\label{sec:sim1}
The first setting is designed to show the advantages of \methodname
relative to alternative methods, especially \superPC. We note that
other methods that employ screening by the marginal
correlation~\citep{ding2017predicting,piironen2018iterative} 
will have similar deficiencies. 
Because \superPC works well if Equation \eqref{eq:assumption} holds,
we design $\boldsymbol{\Sigma}$ to violate this condition and
set the first 15 features to have non-zero $\bstar$ but allow only the first 10
features to have non-zero correlation with the phenotype. \rev{This behaviour is
achieved with line 8 of \autoref{alg:synthetic-data}. By solving this equation
in one unknown component of $\Theta$, we force $\Phi=0$ for the third group of 5
components. Thus, as described in \autoref{sec:background},
Equation~\eqref{eq:assumption} will not hold: some $\Phi_j=0$ but $\beta^*_j\neq
0$.}
We set the true dimension of the subspace as $d = 3$, and we use the
correct dimension for methods based on principal components. 

\autoref{fig:reg1} shows the performance of \methodname and
state-of-the-art alternatives. 
In addition to reporting each method's prediction MSE
on the test set, we also show the number of features selected, precision,
recall, and the ROC curve. 
The ISPCA implementation does not select features.
In this example, \methodname actually outperforms the oracle estimator,
attaining smaller MSE while generally selecting the correct features. This
seemingly implausible result is likely because the variance
of estimating OLS on 15 features is large relative to that of
estimating the low-dimensional manifold followed by 3 regression
coefficients. \methodname has a clear advantage over all the alternative
methods, especially \superPC which is 3 orders of magnitude
worse. \superPC works so poorly because it ignores 5 features. 
ISPCA has slightly lower MSE than SPC.
Ridge is the worst, due to
fitting a dense model when a sparse model generated the data.
\methodname reduces MSE significantly relative to simply using FPS due to more accurate feature selection.
The right plot in \autoref{fig:reg1} further shows the ROC curve for
\methodname, Lasso, Elastic Net, \superPC and AIMER in which we can easily vary
the tuning parameter and select various numbers of
features. \methodname and AIMER have a perfect ROC curve, while the
other 3 methods are unable to identify 5 features.
We undertake a similar exercise under conditions favorable to \superPC in the
Supplement.

\subsection{Semi-synthetic analysis with real genomics data}
\label{sec:semi}

The simulations in \autoref{sec:sim} explore various scenarios for the
data generation process and show the performance of \methodname
relative to the alternatives, however, they do not use any real
genomic data. In this section, rather than fully generating
$\mathbf{X}$, we create a semi-synthetic analysis wherein only the phenotypes are generated.
We first performed PCA on the NSCLC data
\citep{lazar2013integrated} and note that the first two empirical
eigenvalues are relatively large, so we chose the number of PCs to be $d=2$. 
We keep the top 20 rows in the empirical $\mathbf{V}$ which have the
largest norm and set the rest to 0. We then recombine and add
noise. The phenotype is constructed as in the previous simulations, and
the SNR is calibrated as in \autoref{sec:sim1}. 
\autoref{fig:semi1} shows the results analogous to those in \autoref{fig:reg1}. 
\methodname continues to perform well relative to alternatives, though
here, FPS has similar MSE, albeit poor feature selection.

\subsection{Analysis of real genomics data}
\label{sec:real}

We analyze 5 microarray datasets that are publicly available and widely used as
benchmarks. Four of the datasets present messenger RNA (mRNA) abundance
measurements from patients with breast cancer \citep{van2002gene,
  miller2005expression}, diffuse large B-cell lymphoma (DLBCL)
\citep{rosenwald2002use}, and acute myeloid leukemia (AML)
\citep{bullinger2004gene}, and the fifth reports microRNA (miRNA) levels from
non-small cell lung cancer (NSCLC) patients \citep{lazar2013integrated}. The
features in $\X$ are gene expression measurements from microarrays.  \rev{In the
Supplement, we apply \methodname to predict COVID-19 viral load from RNA-Seq
data. }

The phenotypes $Y$ are censored survival time in all cases, though some
of the datasets also contain binary survival status indicators. 
Because the real valued phenotype is non-negative and right censored, we follow common practice and transform $Y$ to $\log(Y + 1)$. 
Each observation is a unique patient. 
The first breast cancer dataset has 78 observations and 4751 genes, the second
has 253 observations and 11331 genes, DLBCL has 240 observations and 7399 genes,
AML has 116 observations and 6283 genes, and NSCLC has 123 observations and 939
genes.


We randomly split each dataset into 3 folds for training, validation, and
testing with proportions $40\%$, $30\%$ and $30\%$ respectively.  
We set the number of components $d=3$ and search over 5 log-linearly spaced
$\lambda$ values. Other choices for $d$ and $\lambda$ yield similar results. 
We train all methods on the training set, use the validation set to
choose any necessary tuning parameters, and report performance of each
method on the test set. We repeat the entire process (data splitting,
validation, and testing) 10 times to reduce any bias induced by the
random splits. In all cases, all methods were tuned to optimize
validation-set MSE.

\autoref{tab:reg} shows the average prediction MSE and the average
number of selected features for \methodname and any  alternative methods that perform feature selection. 
\methodname works better than all the alternative methods on 4 out of
5 datasets with a comparatively small number of features selected. The
DLBCL data is difficult for 
both sparse and PC-based methods.
\rev{As described in \autoref{sec:method}, FPS cannot be used for these data sets
because of the number of genes.}
Non-sparse alternatives have much smaller MSE, suggesting
that many genes may play a roll in mortality rather than only a subset.
SPCA is designed to maximize the variance explained by the principal
components subject to a penalty on the non-sparsity, and it does not
seem to work well in regression tasks.
DSPCA has relatively low prediction MSE, and it does in principle perform feature selection, though it generally produces a dense model.
While Ridge, Random Forests and SVM predict well in general, they do not perform
any feature selection, which is a key objective here, so  
show their MSE in the Supplement.

To assess the potential relevance of the genes selected by \methodname to the
cancer type from which they 
were identified, we further explored the DLBCL data and extracted
the selected genes. (We do the same with AML in the Supplement.)
We first find the best $\lambda$ via 5-fold cross-validation on all the data and
then train \methodname with this $\lambda$.  
Our model selects 87 features corresponding to 32 unique genes and 2
expressed sequence tags (ESTs) for DLBCL.
%
%
Seventeen of the identified genes encode ribosomal proteins,
overexpression of which is associated with poor prognosis
\citep{rpgdlbcl}. A further 9 genes encoding major histocompatibility
complex class II (MHCII) proteins were detected, a notable finding in
light of the fact that MHCII downregulation is a means by which some
DLBCLs evade the immune system \citep{dlbclmhcii}. Discovering these large groups
of similarly functioning genes illustrates the benefits of \methodname
relative to alternatives. \textit{CORO1A}
encodes the actin-binding tumor suppressor p57/coronin-1a, the
promoter of which is often hypermethylated, and therefore likely
silenced in DLBCL \citep{dlbclp57}. \textit{FEZ1} expression has been
used in a prognostic model \citep{dlbclfez1}. \textit{RAG1},
encoding a protein involved in generating antibody diversity, can
induce specific genetic aberrations found in DLBCL
\citep{dlbclnrco}. \textit{RYK} encodes a catalytically dead receptor
tyrosine kinase involved in Wnt signaling and \textit{CXCL5} encodes a
chemokine. To our knowledge, neither gene has been implicated in DLBCL
and thus may be of interest for further exploration. EST Hs.22635
(GenBank accession AA262469) corresponds to a portion of
\textit{ZBTB44}, which encodes an uncharacterized transcriptional
repressor, while EST Hs.343870 (GenBank accession AA804270) does not
appear to be contained within an annotated gene. 
The Supplement lists the selected genes and associated
references. A separate listing of the genes encoding ribosomal and MHCII
proteins are given in the Supplement.

\begin{table*}[tb!]
  \centering
  \resizebox{\textwidth}{!}{
    \begin{tabular}{@{} l r r c r r c r r c r r c r r @{}}
    \toprule
    & \multicolumn{2}{c}{Breast Cancer1}  &\phantom{}&        \multicolumn{2}{c}{Breast Cancer2}  &\phantom{}& \multicolumn{2}{c}{DLBCL} &\phantom{}&  \multicolumn{2}{c}{AML}  &\phantom{}&   \multicolumn{2}{c}{NSCLC}  \\
    \cline{2-3} \cline{5-6} \cline{8-9} \cline{11-12} \cline{14-15}
    Method  &  MSE  & feature\# &&  MSE  & feature\#   &&  MSE  & feature\#  && MSE  & feature\#  && MSE  & feature\#  \\
    \midrule
    \methodname  & \textbf{0.5980}   & 80  && \textbf{0.4168}  & 121  && 0.7073  & 48    && \textbf{1.9568}   & 75   && \textbf{0.1970}  & 27 \\
    Lasso   & 0.7141   & 7  && 0.4622  & 39 && 0.6992  & 31   && 2.0998   & 3  && 0.2263  & 4 \\
    ElasticNet  & 0.6845  & 41 && 0.4517 & 104 && \textbf{0.6869}  & 87 && 2.0820  & 5   && 0.2332 & 20 \\
    \superPC  & 0.6188  & 59 && 0.4179 & 823 && 0.7677 & 67 && 2.3237 & 62   && 0.2795 & 62 \\
    ISPCA  & 0.8647 & NA  && 0.5882 & NA && 0.9441  & NA && 2.3109 & NA && 0.2408 &  NA \\
    AIMER  & 0.6629  & 76 && 0.4192 & 795 && 0.7003  & 76 && 1.9737  & 36 && 0.2120 &  50 \\
    SPCA  & 17.0965 & 212 && 4.7239 & 38 && 2.5980 & 652 && 31.11 & 1043 && 0.9757 & 387  \\
    DSPCA  & 0.6132 & 4374 && 0.4557 & 7880 && 0.7249 & 1342 && 1.9781 & 2742 && 0.2041 & 305  \\
    \bottomrule
    \end{tabular}
    }
    \caption{Prediction MSE and number of selected features for regression of survival time on gene expression measurements}
    \label{tab:reg}
\end{table*}

\section{Theoretical guarantees}
\label{sec:theory}

When the sparse factor model described in \autoref{sec:empirical} is true,
\methodname enjoys near-optimal convergence rates. We now make the necessary
assumptions concrete and note that some can be weakened.

\begin{itemize}
    \item[A1.] $Y_i=X_i^\top \bstar + \epsilon_i$, $i=1,\ldots,n$, where
      $\epsilon_i\sim \mbox{N}(0,\sigma_y)$, $\sigma_y>0$. 
    \item[A2.] $X_i \sim \mbox{N}_p(0,\boldsymbol{\Sigma})$, $i=1,\ldots,n$.
    \item[A3.] $\boldsymbol{\Sigma}=\mathbf{V L V}^\top$, is symmetric,
      $\mathbf{V}^\top \mathbf{V}=\mathbf{I}_p$, $\mathbf{L}$ is diagonal. 
    \item[A4.] $l_i = \lambda_i^2 \mathbf{1}(i\leq d) + \sigma_x^2$ and
      $\lambda_1-\lambda_d := \phi > 0$.
    \item[A5.] $\snorm{\mbox{diag}\left(\mathbf{V}_d\mathbf{V}_d^\top\right)}_0
      \leq s$ and  
      $\min_j \left\{(\mathbf{V}_d\mathbf{V}_d^\top)_{jj} \vee 0\right\} >
      2\tau$. 
    \item[A6.] \rev{as $n,p\rightarrow\infty$,  $n > (s^2+ d)\log(p)$ eventually.}
    
\end{itemize}
\rev{Assumptions A1--A4 are the same as those used in \autoref{sec:sim} to
  generate data from a linear
  factor model. Assumption A5 says that the number of true nonzero coefficients
  $\bstar$ must be no more than $s$ and that the size of the associated components
  must be large enough. Assumption A6 means that eventually, we must have at
  least as many 
  observations $n$ as a logarithmic function of $p$ times the true number of
  components plus the square of the number of nonzero $\bstar$ coefficients.}
\begin{theorem}
  \label{thm:thm1}
Suppose assumptions A1 to A6 hold and let $\hat{\beta}$ be the
estimate produced by \textnormal{\methodname} with $\lambda = c
\lambda_1\sqrt{\log(p)/n}$ and $t < 2\tau$ \rev{where $t$ is the threshold
  used in \autoref{alg:suffpcr} and $\tau$ is given in A5}. Then
\[
\frac{1}{n}\snorm{\mathbf{X}(\hat\beta - \bstar)}^2_2 = \mathcal{O}_P
\left(\frac{(s^2+d)\log(p)}{n}\right). 
\]
\end{theorem}

\begin{theorem}
    \label{thm:thm2}
Suppose assumptions A1 to A6 hold and  let $\hat{\beta}$ be the
estimate produced by \textnormal{\methodname} with $\lambda = c
\lambda_1\sqrt{\log(p)/n}$ and $2\tau > t> \tau$ \rev{where $t$ is the threshold
used in \autoref{alg:suffpcr} and $\tau$ is given in A5}. Then
\[
\left|\mathrm{supp}\left(\hat\beta\right) \bigtriangleup
  \mathrm{supp}\left(\bstar\right)\right| = \mathcal{O}_P\left(  
\frac{s^2 \log(p)}{n}
\right),
\]
where $A \bigtriangleup B= A/B\cup B/A$ is the symmetric difference operator and $\mathrm{supp}$ denotes the support set.
\end{theorem}

\rev{In both results above, $c$ is a positive number (possibly different between
  the two) that is independent of $n$ and $p$, but may depend on any of the
  other values given in A1--A6.} \autoref{thm:thm1} gives a convergence rate for
the prediction error of 
\methodname comparable to that of Lasso 
though with explicit additional dependence on $d$. Under standard
assumptions with fixed design, this dependence would not
exist for Lasso. On the other hand, our results are for random design
with $d$ small, along with different constants absorbed by the
big-$\mathcal{O}$.
\autoref{thm:thm2} shows that
our procedure can correctly recover the set of nonzero $\beta^*$ as long as the
threshold $t$ is chosen correctly. We note that this result is a
direct consequence of \citet[Theorem 3.2]{vu2013fantope}. In
practice, the condition $2\tau>t>\tau$ cannot be verified, although
the ``elbow'' condition we employ in the empirical examples seems to
work well. Finally, we emphasize
that, as is standard in the literature,  these results are
for asymptotically optimal tuning parameters $\lambda,\ t$ rather
than empirically chosen values.
The proof of \autoref{thm:thm1} is given in the Supplement. These
results suggest that \methodname is nearly optimal as $p$ and
$n$ grow.

\section{Discussion}
\label{sec:con}

High-dimensional prediction methods, including regression and classification,
are widely used to gain biological insights from large datasets. Three main goals in this 
setting are accurate prediction, feature selection, and computational tractability.
We propose a new method called \methodname which is capable of achieving these goals
simultaneously. \methodname is a linear predictor on estimated sparse
principal components. Because of the sparsity of the projected 
subspace, \methodname usually selects a small number of features.
We conduct a series of synthetic, semi-synthetic and real data analyses to 
demonstrate the performance of \methodname and compare 
it with existing techniques. We also prove 
near-optimal convergence rates of \methodname under sparse assumptions.
\methodname works better than alternative methods when the 
true model only involves a subset of features.

\section*{Funding}

The authors gratefully acknowledge support National Science Foundation (grant
DMS–1753171 to DJM) the National Institutes of Health (grant R35GM128631 to GEZ)
and the National Sciences and Engineering Research Council of Canada (NSERC)
(grant RGPIN-2021-02618 to DJM).


\bibliographystyle{plainnat}
\bibliography{ref}

\clearpage

\appendix

 \section{\rev{Genetics data discussion}}

 \rev{In this section we provide more details on the 5 standard genetics data sets
 discussed in the manuscript. We then apply our methodology to one/two
 additional data sets which are less well studied in the literature.}

\subsection{\rev{Detailed description}}

\rev{
  In the manuscript (and below in \autoref{sec:real_more}), we analyze 5
  microarray datasets that are publicly available and widely used as benchmarks.
  Four of the datasets present messenger RNA (mRNA) abundance measurements from
  patients with breast cancer \citep{van2002gene, miller2005expression}, diffuse
  large B-cell lymphoma (DLBCL) \citep{rosenwald2002use}, and acute myeloid
  leukemia (AML) \citep{bullinger2004gene}, and the fifth reports microRNA
  (miRNA) levels from non-small cell lung cancer (NSCLC) patients
  \citep{lazar2013integrated}.
}

\rev{
  These data sets have
  between about 80 and 250 patients with expression measurements on 900 to 11000 genes.
  \autoref{tab:dataset} gives specific statistics about each dataset.
}

\begin{table}[h]
  \color{Black}
  \centering
  \begin{tabular}{@{}lrr@{}}
    \toprule
    Name & $n$ (patients) & $p$ (genes) \\
    \midrule
    Breast cancer \citep{van2002gene} & 78 & 4751 \\
    Breast cancer \citep{miller2005expression} & 253 & 11331 \\
    DLBCL \citep{rosenwald2002use} & 240 & 7399 \\
    AML  \citep{bullinger2004gene} & 116 & 6283 \\
    NSCLC \citep{lazar2013integrated} & 123 & 939 \\ 
    \bottomrule
  \end{tabular}
  \caption{Summary of canonical datasets}
  \label{tab:dataset}
\end{table}

\rev{
  The AML data was originally collected and analyzed by
  \citep{bullinger2004gene}. They used complementary-DNA microarrays to measure
  gene expression from either peripheral-blood or bone marrow samples from 116
  adults with AML. The gene expression measurements for the 6283 genes that were
  highly variable across patients are included. The observed outcomes are the
  (possibly right-censored) survival time in days as well as a binary indicator
  for whether or not the patient died.
}

 \rev{ 
  The first set of breast cancer data from \citep{van2002gene} is based on 78
  sporadic lymph-node-negative patients. The authors derived cRNA from
  snap-frozen tumor samples and pooled across each of the sporadic carcinomas.
  The 4751 available genes were significantly regulated across the patients
  (showing at least a two-fold difference across at least 5 patients). The
  outcome is the follow-up time (in months) for metastases as well as a binary
  indicator for whether prognosis was ``good'' (no metastases within 5 years) or
  ``poor''. 
}

\rev{
  The second set of breast cancer gene expression measurements is based on 253
  patients who's cancer tissues were sequenced for the p53 mutation
  \citep{miller2005expression}. A binary variable for the presence of the p53
  mutation is missing for two patients, which were removed from further
  analysis. The regression exercise focuses on the right-censored survival time
  (measured in years). The measured gene expressions come from Affymetrix
  high-density oligonucleotide arrays. A number of other clinical variables are
  also present in the data contained in the R package
  (\url{https://github.com/dajmcdon/suffpcr}). 
}

\rev{
  The DLBCL data was collected by \citep{bullinger2004gene} and contains Biopsy
  samples from 240 patients with gene expression measurements from DNA
  microarrays. The microarrays were constructed from 12,196 clones of cDNA then
  used to quantify the expression of mRNA in the tumors. Genes with significant
  differential expression across patients were included. The response is
  survival time after chemotherapy in years as well as binary survival status
  (57\% of patients died).
}

\rev{Finally, the NSCLC dataset was analysed by \citep{lazar2013integrated}. The
  expression measurements consist of from paired (tumor and non-tumor) micro RNA
  collected from 123 cancer patients. The response variable is the time to
  relapse in years. Other clinical variables are also available.}

\section{\rev{Application to COVID-19 Data}}

\begin{figure}[tb!]
  \centering
  \includegraphics[width=.7\linewidth]{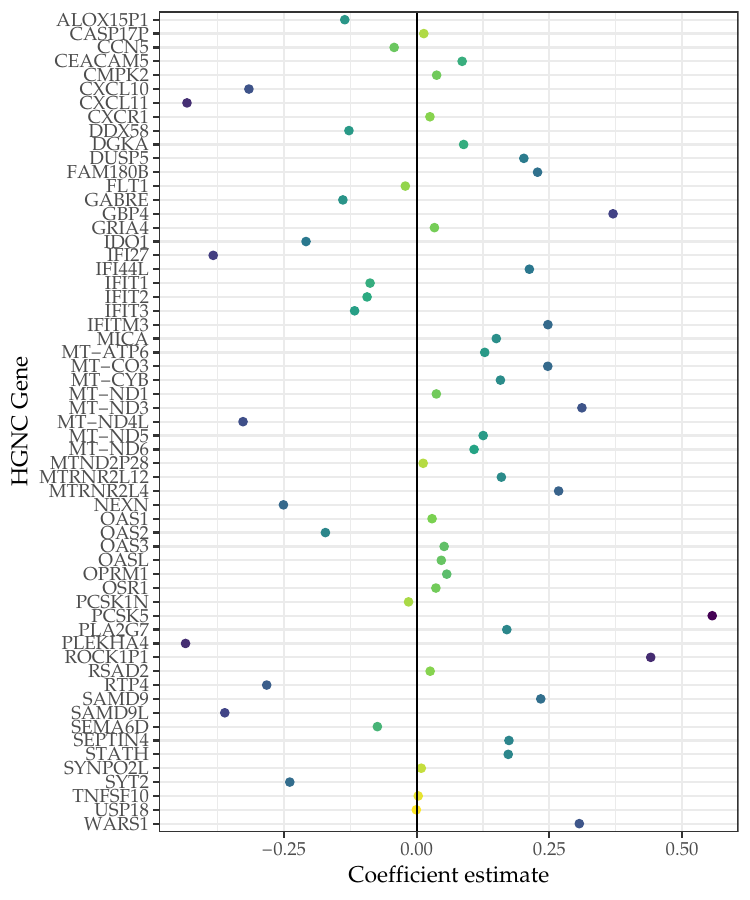}
  \caption{\rev{Coefficient estimates from \methodname on RNA-Seq measurements of
    COVID-19-positive patients. The colour corresponds to the magnitude.}}
  \label{fig:covid-19}
\end{figure}

\rev{In this section, we apply our methods to RNA-Seq data for COVID-19-positive
  patients. As a developing disease, biological analysis remains highly
  desirable, and so to facilitate this sort of undertaking, recent work has
  endeavoured to collect and distribute this data publicly \citep{Zhang2021}.}

\rev{Here, we examine the data collected by \citet{Lieberman2020}, and examine
  how well RNA-Seq measurements predict viral load. The data contains results
  for 430 PCR-confirmed COVID-positive patients and 54 negative controls.
  \citet{Lieberman2020} examines both differential expression between positives
  and 
  controls and the viral load of the positive patients, while we focus on only
  the viral load. Of the 430 confirmed positives, 413 have reported usable
  measurements. Viral load in this case is measured through a proxy: the cycle
  threshold (Ct) of the SARS-CoV-2 nucleocapsid gene region 1 (N1) target from
  a PCR test. Larger Ct values indicate more cycles are required to detect the
  N1 target, and thus, likely indicate \emph{lower} viral load. While
  \citet{Lieberman2020} bin this continuous measurement into ``high'' and
  ``low'' groups to undertake differential analysis, our methodology directly
  models the continuous response. This allows for (1) increased ability to
  detect potentially predictive genes and (2) direct quantification of the
  effect of increased expression on viral load.}

\rev{The raw expression measurements are preprocessed before being used in
  \methodname. First, we remove any genes whose median expression level across
  the 413 patients with Ct measurements is 0. This reduces the data from
  $\sim$36,000 genes to 9435. We then transform the raw counts using
  $\tilde{z}\mapsto\log(\tilde{z}+1)$. Finally, we centre and scale by the mean
  and standard deviation of
  the similarly transformed measurements from the healthy controls.
  Mathematically, if $\tilde{z}_{gt}$ is the vector of expression measurements
  for gene $g$ in the treatment group and $\tilde{z}_{gc}$ is the vector of
  expression measurements for the same gene in the control group, we form
  $$x_{gt} = \frac{\log(\tilde{z}_{gt}+1) -
    \overline{z}_{gc}}{\textrm{sd}(z_{gc})},$$ 
  where $\overline{z}_{gc} = \textrm{mean}(\log(\tilde{z}_{gc}+1))$ and
  $\textrm{sd}(z_{gc})=\textrm{sd}(\log(\tilde{z}_{gc}+1))$.
} 

\rev{
  We apply \methodname to the $\X$ matrix formed by the columnwise concatenation
  of $x_{gt}$ with the N1 Ct as the response vector. We examined embedding sizes
  of $d\in\{3,\ 5,\ 15\}$ and $\lambda$ on a $\log_{10}$-spaced grid between 0
  and 1. We chose both parameters using the minimum of the 5-fold
  cross-validation error. The result is shown in \autoref{fig:covid-19} (here
  $d=15$). Our method selected 59 genes. Many of these with the largest
  magnitude (darkest colour) are similar to those described in
  \citep{Lieberman2020}: the CXCL10 and 11 genes along with IDO1 and IFI27 are
  proinflammatory and/or interferon induced and may be related to the
  ``cytokine storm'' found in some patients \citep{cxcl10}. Novel and
  potentially interesting PLEKHA4, which is strongly related to higher viral
  load, but more closely related to melanomas. PCS5K is more strongly expressed
  in patients with lower viral load. This gene encodes a proprotein convertase
  that may potentially help to clear the virus. ROCK1P1, also more strongly
  expressed in lower viral load patients, is not well understood.  
}

\begin{figure}[tb!]
  \centering
  \includegraphics[width=.7\linewidth]{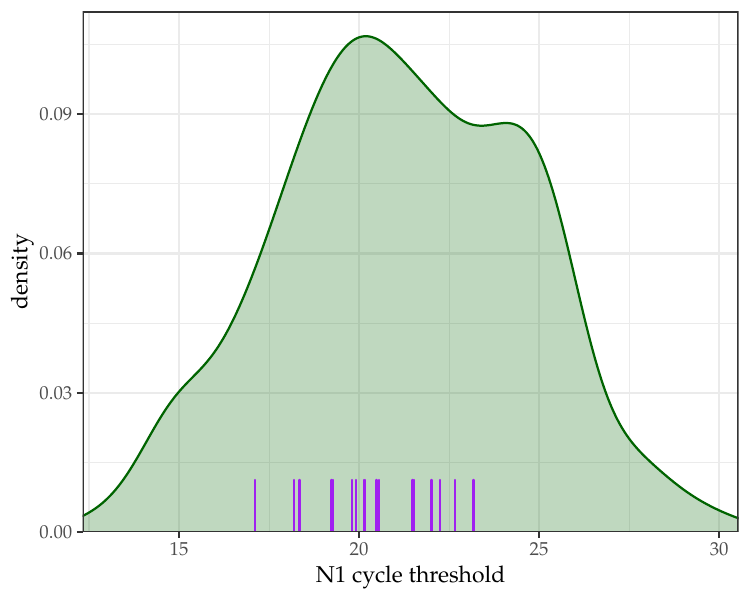}
  \caption{\rev{Density of observed N1 Ct values for 413 patients (green area) and
    predicted values for 27 positive patients with ``Unknown'' values (purple
    lines) based on \methodname estimated from the complete cases. }}
  \label{fig:covid-predictions}
\end{figure}

\rev{
  As mentioned above, 27 patients were positive but did not have a measured N1
  Ct from the PCR test. These were also missing age and gender information.
  Using the fitted \methodname model, we can predict their missing N1 cycle
  thresholds. These predictions are shown in \autoref{fig:covid-predictions}.
  The green area shows a density estimate for the 413 patients with observed
  measurements while the purple lines display predicted values for those
  positive patients whose N1 Ct values are labeled ``Unknown''. Because the data
  is missing, the accuracy of these predictions cannot be determined.
}
 
\section{Additional experiments and results for regression}

\subsection{Conditions favorable to \superPC}

This example is designed to show the performance of \methodname under favorable
conditions for \superPC.  
Here, the only alteration to the data generation process is that we set the
first 10 features to have non-zero 
$\bstar$, the same features which have non-zero marginal
correlations with the phenotype.
\rev{This is achieved using Algorithm 3 in the
manuscript with a few alterations:
(1) in line 5, we generate $\widetilde{\mathbf{V}}_d \in \R^{(d-1)\times 3}$
  (recall that we have chosen $d=3$); and (2) we replace lines 7--9 by directly simulating $\Theta
  \in \R^d$ with i.i.d.\ 
  standard normal entries.}

\autoref{fig:reg2} gives the analogous results for this example.
Here, \methodname has comparable MSE to \superPC, which is the
best, slightly better than the oracle as was the case for \methodname
above. However,
\superPC is less likely to select the correct features,
while \methodname 
selects the correct features most of the time. 
Furthermore, \methodname tends to select fewer features, and it has the best precision and recall (Ridge will always have recall equal to 1). 

\begin{figure*}[tb]
  \centering
  \includegraphics[width=0.9\linewidth]{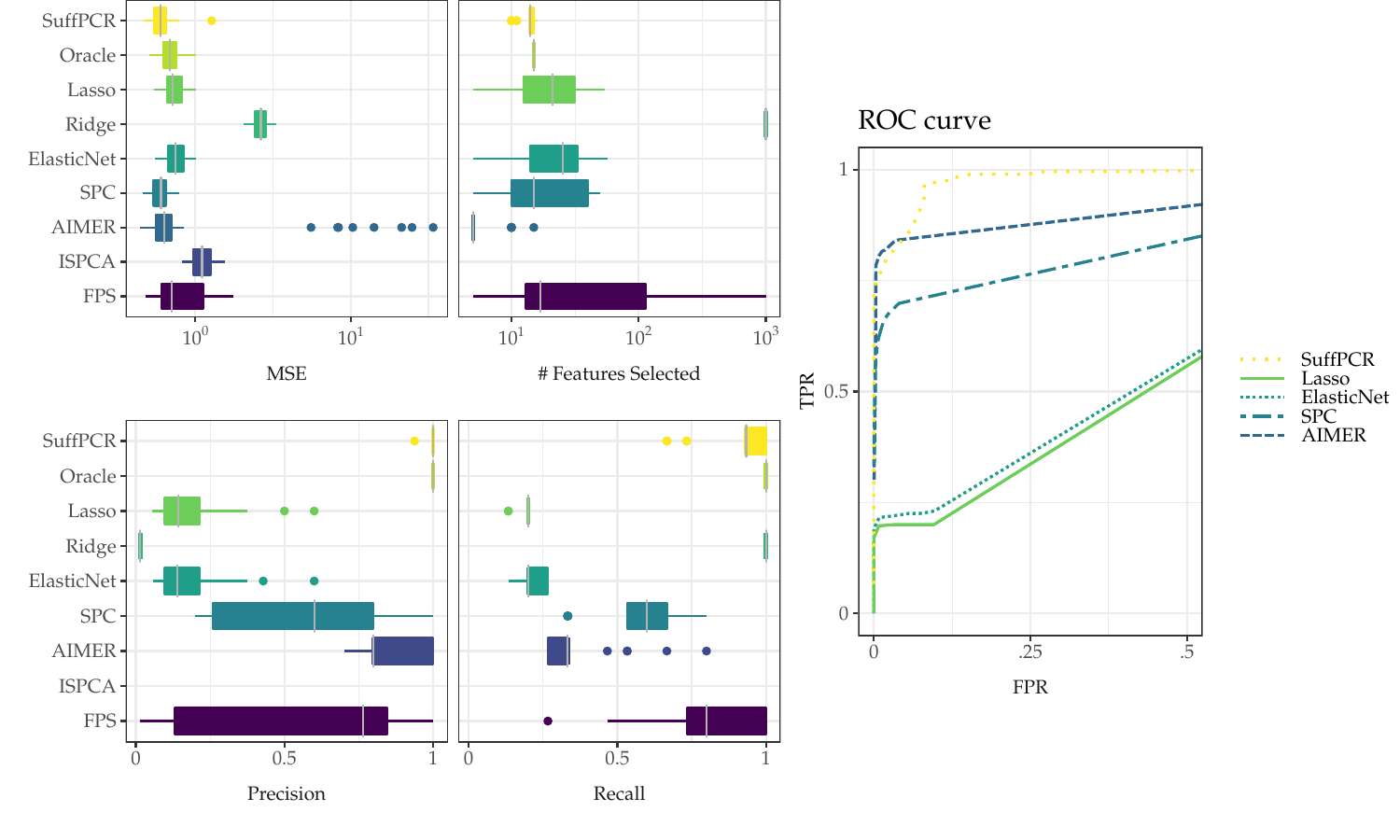}
  \caption{This figure demonstrates the performance of \methodname under
    favorable conditions for \superPC.  We have 
    omitted the other methods from the ROC curve for legibility, but their
    behavior is similar to lasso.}
  \label{fig:reg2}
\end{figure*}

\begin{figure}[tbh]
  \centering
  \includegraphics[width=0.9\linewidth]{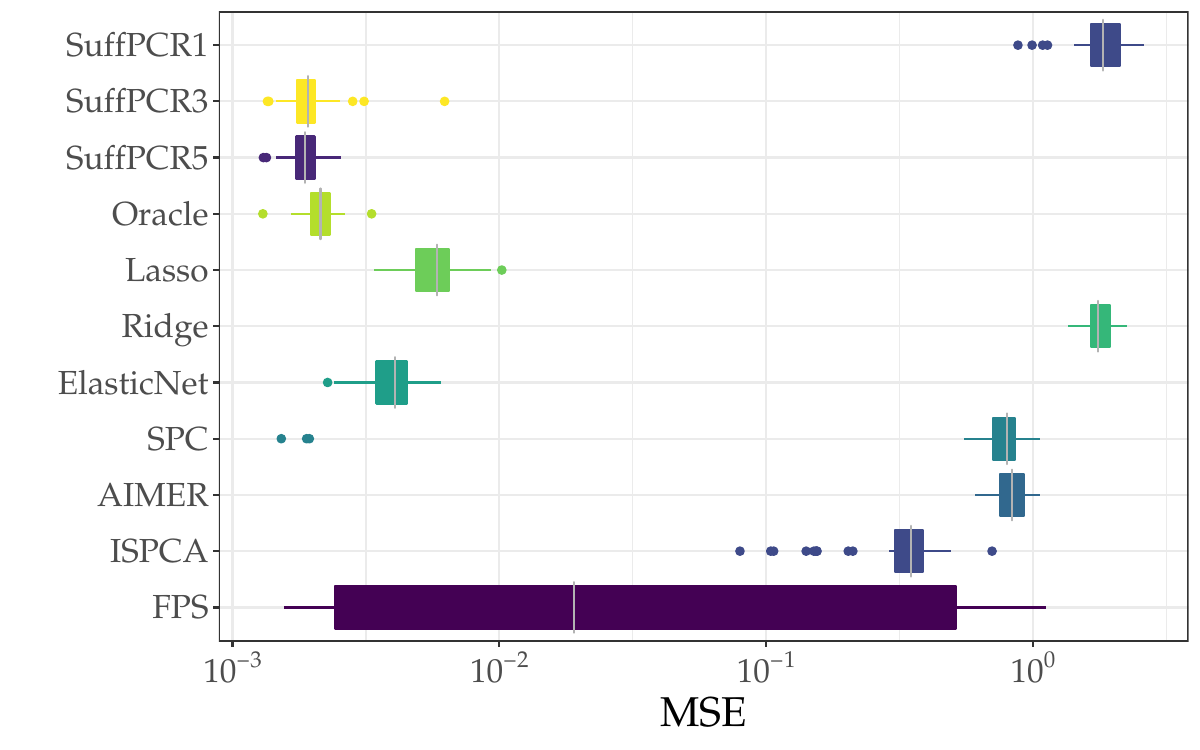}
  \caption{This figure compares the prediction MSE of \methodname
    with alternatives with different $d$ in \methodname.}
  \label{fig:reg4}
\end{figure}


\subsection{Selecting tuning parameters}

In \methodname, the tuning parameters are the penalty term $\lambda$,  the dimension of the projected subspace $d$, and the threshold $t$. 
In all the simulations, we examine the same values of $\lambda$.
When analyzing real datasets, our algorithm offers automatic calculation of potential values of $\lambda$ given the sample covariance matrix.
In practice, the more penalty parameters we explore, the better the results could be assuming the sample size is large enough that risk estimation is not too volatile.

To demonstrate the importance of selecting $d$, we simulate the data with $d = 3$ and set all the other parameters the same as in the first simulation. 
We estimate \methodname using $d\in \{1,\ 3,\ 5\}$ and display the results in
\autoref{fig:reg4}. 
When $d$ is smaller than the true value, it is hard for \methodname to capture all the information contained in $\X$, thus \methodname has terrible prediction MSE, worse than any other methods.
When $d$ is larger than the true value, \methodname remains reasonable because its bias is small, but the variance will eventually increase as $d$ increases, diminishing performance.
Selecting $t$ is not difficult in our simulations, so the results have been
omitted, but it is less trivial in the real data settings.

    
    
    

 \subsection{Investigation of signal-to-noise ratio with synthetic data}

In all or our synthetic examples, the data generating model has two sources of noise, one from constructing $\X$ corresponding to $\sigma_x$, and the other from constructing $Y$, denoted $\sigma_y$. 
This simulation aims to control the two noise sources separately to examine their effect on the performance of \methodname. 
Note that $\mathbf{X}$ is random unlike in standard
prediction studies, so both sources of
noise are important. We use $\textrm{SNR}_x$ and $\textrm{SNR}_y$ to denote the signal-to-noise ratio for $\X$ and $Y$ respectively. 
These are given by
\begin{align}
    \textrm{SNR}_x 
    &= \frac{\Expect{\norm{ \mathbf{U}_d \mathbf{\Lambda}_d \mathbf{V}_d^\top }_F}}{\Expect{\norm{ \sigma_x \mathbf{E}}_F }}  
    = \sqrt{\frac{ \mbox{tr}(\Lambda_d^2)}{ p\sigma_x^2}}, \\
    \textrm{SNR}_y 
    &= \frac{\Expect{ \norm{\X \bstar }_2}}{ \Expect{ \norm{\sigma_y Z}_2 } }
    = \sqrt{ \frac{ \beta^{*\top} \mathbf{V}_d \mathbf{\Lambda_d^2} \mathbf{V_d^\top} \bstar + \sigma_x^2 \lVert \bstar \rVert_2^2 }{  n \sigma_y^2 } }.
\end{align}
Note that $\textrm{SNR}_y$ depends not only on $\bstar$ but also on $\sigma_x$
and the linear manifold through $\mathbf{V}_d$.

We alter the values of $\textrm{SNR}_x$ and $\textrm{SNR}_y$
to generate $\X$ and $Y$ while everything else is as in the first simulation.
\autoref{fig:reg3} shows the prediction MSE for all methods on four sets of simulated data for both high and moderate combinations of  $\textrm{SNR}_x$ and $\textrm{SNR}_y$. 
When the SNR decreases, the prediction MSE of all the methods
increases as expected. In all configurations, \methodname performs
similarly to the oracle, better than Lasso or Elastic Net, while
\superPC is more similar to Ridge 
regression (though better). 
Interestingly, changing $\textrm{SNR}_x$ and $\textrm{SNR}_y$ has a similar
impact on nearly all the methods except Ridge, \superPC, AIMER, and ISPCA, which are nearly unaffected and uniformly worse by an order of magnitude.

\begin{figure*}[t]
    \centering
    \includegraphics[width=0.9\linewidth]{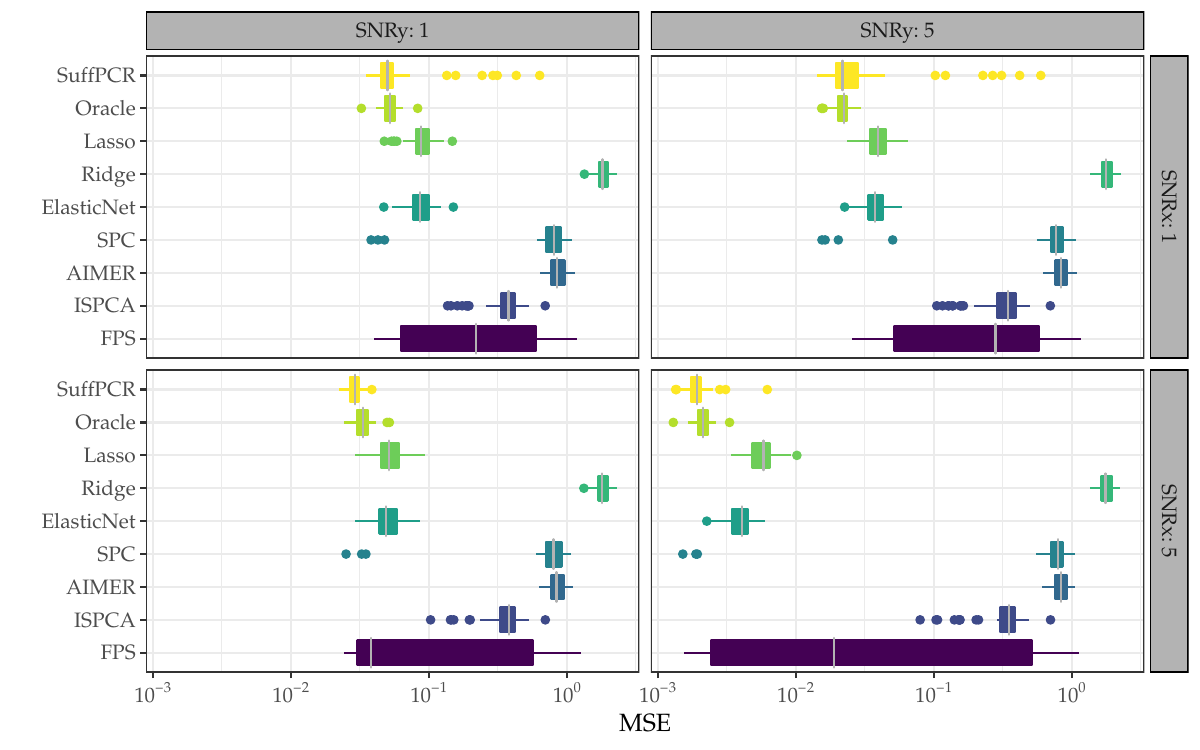}
    \caption{This figure compares the prediction MSE of \methodname with
      alternatives across different values of $\textrm{SNR}_x$ and
      $\textrm{SNR}_y$. The $x$-axis is on the log scale. Boxplots are over 50
      replications. } 
    \label{fig:reg3}
\end{figure*}

\subsection{Predictive genes/miRNA for additional cancers}
\label{sec:real_more}

\autoref{tab:DLBCLgenes} enumerates genes selected by \methodname on the DLBCL data,
\autoref{tab:AMLgenes} lists the selected genes for AML, 
\autoref{tab:NSCLCgenes} lists the miRNA sequences selected by \methodname on the NSCLC data, while \autoref{tab:BreastCancer1genes} lists all the features selected by \methodname on the Breast Cancer1 data.
\autoref{tab:reg2} gives the prediction MSE for Ridge, Random Forests
and SVM for the 5 regression datasets we examine in Section 3.3 of the
manuscript.
These methods predict well in general, but they do not select
features.

For AML we describe the related work summarized in \autoref{tab:AMLgenes} in
more detail. Of the 4 discovered genes,  3 have
been discussed in the literature. \textit{BPGM}, encoding a
multifunctional metabolic enzyme restricted to erythrocytes and
placental cells, is upregulated in mouse models of AML
\citep{novak2012gene}. \textit{PI4KB}, encoding a lipid kinase, is
amplified in breast cancer \citep{waugh2014pi4kb} and has been
proposed to be a druggable AML-specific dependency
\citep{Zhou_2020}. \textit{LOC90379} encodes an uncharacterized
protein that is highly immunogenic in ovarian cancer serum
\citep{Gnjatic5088}. The fourth feature discovered, Hs.321434, is an
EST (GenBank accession H96229) that corresponds to an intronic
sequence of an uncharacterized long noncoding RNA (lncRNA),
LOC101929579, and is thus of uncertain significance. 
Similar listings for the discovered genes for the NSCLC and Breast
Cancer 1 data are given in the Supplement without the accompanying
literature review.

Finally, \autoref{tab:dlbcl-proteins} expands on the listing in the
main document, listing the genes encoding ribosomal proteins and MHCII
protein which predict DLBCL survival as selected by \methodname.

\begin{table}
\centering
\begin{tabular}{@{}lcl@{}}
    \toprule
     Gene 
     &  Related ? & Reference\\
    \midrule
    Ribosomal protein genes (17) 
    & \cmark & \citet{rpgdlbcl}\\
    MHCII (9)
    & \cmark & \citet{dlbclmhcii}\\
    CORO1A 
    & \cmark & \citet{dlbclp57}\\
    FEZ1
    & \cmark & \citet{dlbclfez1}\\
    RAG1 
    & \cmark & \cite{dlbclnrco}\\
    RYK 
    & \xmark & \\
    CXCL5 
    & \xmark & \\
    ESTs Hs.22635, Hs.343870
    & \xmark & \\
    \bottomrule
\end{tabular}
\caption{Predictive genes for DLBCL selected by \methodname.}
\label{tab:DLBCLgenes}
\end{table}

\begin{table}
\centering
\begin{tabular}{@{}lcl@{}}
    \toprule
     Gene 
     &  Related ? & Reference\\
    \midrule
    BPGM 
    & \cmark & \cite{novak2012gene}\\
    PI4KB 
    & \cmark & \cite{Zhou_2020}\\
    EST Hs.321434
    & \xmark & \\
    LOC90379 
    & \cmark & \cite{Gnjatic5088}\\
    \bottomrule
\end{tabular}
\caption{Predictive genes for AML selected by \methodname.}
\label{tab:AMLgenes}
\end{table}

\setlength{\tabcolsep}{9pt}
\begin{table}[h]
\centering
\begin{tabular}{@{}llll@{}}
    \toprule
     \# & Sequence & \# & Sequence \\
    \midrule
    1 & hsa-miR-376c&           21 & hsa-miR-409-5p \\             
    2 & hsa-miR-320c&           22 & hcmv-miR-US4   \\           
    3 & hsa-miR-299-3p&         23 & hsa-miR-376a   \\           
    4 & hsa-miR-154&            24 & hsa-miR-1471\\              
    5 & hsa-miR-410&            25 & hsa-miR-411   \\            
    6 & hsa-miR-1182&           26 & bkv-miR-B1-5p  \\           
    7 & hsa-miR-136&            27 &  hsa-miR-377   \\           
    8 & hsa-miR-379&            28 & hsa-miR-601\\               
    9 & hsa-miR-765&            29 & hsa-miR-299-5p \\           
    10 & hsa-miR-610&           30 & hsa-miR-543    \\           
    11 & hsa-miR-487a &         31 & hsa-miR-381    \\           
    12 & hsa-miR-136*&          32 & hsa-miR-329\\               
    13 & hsv1-miR-H1  &         33 & hsa-miR-760    \\           
    14 & hsa-miR-622  &         34 & hsa-miR-409-3p \\           
    15 & hsa-miR-659  &         35 & hsa-miR-617    \\           
    16 & hsa-miR-376b&          36 & hsa-miR-758\\               
    17 & hsa-miR-154*  &        37 & hsa-miR-1183   \\           
    18 & hsa-miR-483-5p &       38 & hsa-miR-671-5p \\           
    19 & hsa-miR-337-5p &       39 & hsa-miR-127-3p\\            
    20 & hsa-miR-376a*\\
    \bottomrule
\end{tabular}
\caption{RNA sequences selected by \methodname for NSCLC data.}
\label{tab:NSCLCgenes}
\end{table}

\begin{table}[h]
\centering
\begin{tabular}{@{}llll@{}}
    \toprule
     \# & Feature & \# & Feature \\
    \midrule
    1 & NM\_003118 & 15 & NM\_003247  \\    
    2 & Contig46244\_RC & 16 & NM\_004079\\      
    3 & Contig55801\_RC & 17 & NM\_002775\\      
    4 & M37033 & 18 & NM\_004369     \\ 
    5 & NM\_004385 & 19 & NM\_002985 \\    
    6 & Contig43613\_RC & 20 & Contig43833\_RC \\
    7 & Contig42919\_RC & 21 & NM\_005565\\
    8 & NM\_016081 & 22 & Contig52398\_RC \\
    9 & NM\_016187 & 23 & NM\_006889    \\  
    10 & Contig30260\_RC & 24 & Contig66347   \\
    11 & NM\_000089 & 25 & NM\_000090  \\    
    12 & NM\_000138 & 26 & Contig25362\_RC\\ 
    13 & NM\_000393  & 27 & NM\_000560  \\    
    14 & Contig58512\_RC & 28 & NM\_001387\\
    \bottomrule
\end{tabular}
\caption{Features selected by \methodname for Breast Cancer1 data.}
\label{tab:BreastCancer1genes}
\end{table}

\begin{table*}[tb]
  \centering
  \resizebox{\textwidth}{!}{
    \begin{tabular}{@{} l r r c r r c r r c r r c r r @{}}
    \toprule
    & \multicolumn{2}{c}{Breast Cancer1}  &\phantom{}&        \multicolumn{2}{c}{Breast Cancer2}  &\phantom{}& \multicolumn{2}{c}{DLBCL} &\phantom{}&  \multicolumn{2}{c}{AML}  &\phantom{}&   \multicolumn{2}{c}{NSCLC}  \\
    \cline{2-3} \cline{5-6} \cline{8-9} \cline{11-12} \cline{14-15}
    Method  &  MSE  & feature\# &&  MSE  & feature\#   &&  MSE  & feature\#  && MSE  & feature\#  && MSE  & feature\#  \\
    \midrule
    Ridge   & 0.6331   & 4751  && 0.4330  & 11331 && 0.6766   & 7399   && 2.0451   & 6283   && 0.2105  & 939 \\
    Random Forest  & \textbf{0.5519} & 4751 && \textbf{0.3976} & 11331 && \textbf{0.6613} & 7399 && 2.0276 & 6283 && \textbf{0.2059} & 939  \\
    SVM linear  & 0.6375 & 4751 && 0.4009 & 11331 && 0.7414 & 7399 && 2.5637 & 6283 && 0.2819 & 939  \\
    SVM RBF  & 0.5602 & 4751 && 0.4448 & 11331 && 0.6684 & 7399 && \textbf{2.0106} & 6283 && 0.2080 & 939  \\
    \bottomrule
    \end{tabular}
    }
        \caption{Prediction MSE for alternative methods in real genomics data analysis.}
    \label{tab:reg2}
  \end{table*}

\section{Extension to classification}
\label{sec:classification}

\methodname is easily extended to solve classification tasks using logistic regression (or even other classification methods). 
We use logistic regression as an example to show how \methodname performs for classification tasks.
Note that the algorithm is similar to Algorithm 1 except that step 12 is replaced by the objective of logistic regression.

\subsection{Synthetic data for binary classification}
\label{sec:classification_sim}

We first use a simulation to demonstrate the generalization of  \methodname to solve a binary classification problem. 
Using the same data generating model and parameters as in the favorable scenario in Section 3.1.1 of the manuscript, we generate $Y$ using the logistic function. 
As before, we choose tuning parameters with the validation set, and report the overall prediction accuracy on the test set.

\autoref{fig:cla1} shows the classification accuracy compared to the oracle (logistic regression on the true predictors), logistic lasso, logistic ridge, and logistic FPS (logistic regression on the estimated principal components from FPS). 
\methodname has very high classification accuracy on the test set relative to the alternative methods. We also simulate classification data analogously to the other scenarios discussed in the manuscript, and the results are very similar.

\begin{figure}[h]
  \centering
  \includegraphics[width=0.75\linewidth]{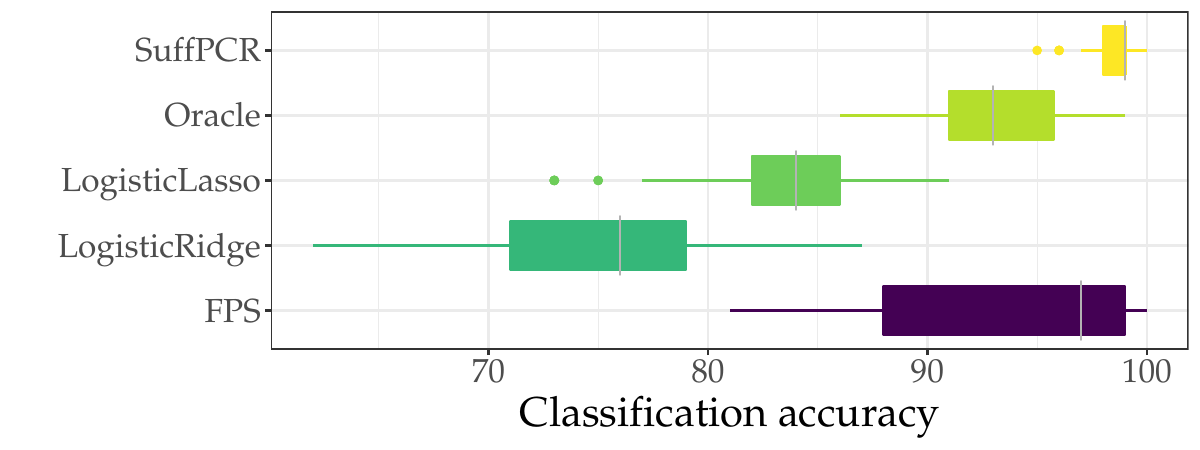}
  \caption{This figure compares the classification accuracy of \methodname with
    alternative methods. The $x$-axis shows the percentage of correct
    classifications on the test set. Boxplots are over 50 repeated simulations.} 
  \label{fig:cla1}
\end{figure}

\subsection{Analysis of real genomics data}
\label{sec:classification_real}

We analyze 3 of the 5 real genomics datasets from Section 3.3 of the manuscript, which include a binary survival status. 
We use the same training test split process as in the manuscript and compare \methodname with logistic lasso and logistic ridge.

\begin{table*}[tbh]
    \centering
    \begin{tabular}{@{}l r r c r r c r r@{}}
    \toprule
    & \multicolumn{2}{c}{Breast cancer1}  &\phantom{}& \multicolumn{2}{c}{DLBCL}  &\phantom{}& \multicolumn{2}{c}{AML}\\
    \cline{2-3} \cline{5-6} \cline{8-9} 
    Method  &  Acc.\%  &  feature \#    &&  Acc.\%  &  feature \#   &&  Acc.\%  &  feature \#  \\
    \midrule
    \methodname   & \textbf{60.95}   & 100  && 60.43   & 103  && \textbf{59.39}   & 95   \\    
    Logistic Lasso   & 60.00   & 7   && 56.86   & 5   && 57.57  & 9  \\
    Logistic Ridge     & 58.09   & 4751   && \textbf{61.00}   & 7399   && 56.97   & 6283   \\
    \bottomrule
    \end{tabular}
        \caption{Accuracy of survival status prediction and number of selected genes for the classification tasks on real data. Results are averaged over 10 repeated iterations.}
    \label{tab:cla}
\end{table*}

\autoref{tab:cla} shows the average classification accuracy and average number of selected genes in the classification tasks.
The results in the classification tasks are similar to those in
regression. Again, the DLBCL dataset tends to be difficult for both
sparse and PC-based methods.

\begin{table}
  \centering
  \begin{tabular}{@{}lll@{}}
    \toprule
    GenBank AN & description\\
    \midrule
    M17886 & ribosomal protein, large, P1\\
    S79522 & ribosomal protein S27a\\
    X03342 & ribosomal protein L32\\
    D23661 & ribosomal protein L37\\
    M84711 & ribosomal protein S3A\\
    L06498 & ribosomal protein S20\\
    U14973 & ribosomal protein S29\\
    M17887 & ribosomal protein, large P2\\
    L04483 & ribosomal protein S21\\
    M64716 & ribosomal protein S25\\
    L19527 & ribosomal protein L27\\
    X66699 & ribosomal protein L37a\\
    U14970 & ribosomal protein S5\\
    X64707 & ribosomal protein L13\\
    U14971 & ribosomal protein S9\\
    M60854 & ribosomal protein S16\\
    NM\_022551 & ribosomal protein S18\\
    X00452 & major histocompatibility complex, class II, DQ alpha 1\\
    X00457 & major histocompatibility complex, class II, DP alpha 1\\
    X62744 & major histocompatibility complex, class II, DM alpha\\
    U15085 & major histocompatibility complex, class II, DM beta\\
    M16276 & major histocompatibility complex, class II, DQ beta 1\\
    K01171 & major histocompatibility complex, class II, DR alpha\\
    M20430 & major histocompatibility complex, class II, DR beta 5\\
    M83664 & major histocompatibility complex, class II, DP beta 1\\
    K01144 & CD74 antigen (invariant polypeptide of major\\
    &\quad histocompatibility complex, class II antigen-associated)\\
    \bottomrule
  \end{tabular}
  \caption{A listing of the genes encoding ribosomal proteins and MHCII
protein which predict DLBCL survival as selected by \methodname.\label{tab:dlbcl-proteins}}
\end{table}

\section{Approximate singular value decomposition in Algorithm 1}

To approximate the top $j$ eigenvalues of a symmetric matrix $\mathbf A$,
we require the $k$-step (partial) Lanczos bidiagonalization $(j<k)$. 
For initial vector $p_1$, this
is given by
\begin{align}
  \mathbf{AP}^{(k)}& = \mathbf{Z}^{(k)}\mathbf{ W}^{(k)}, \;
                 \mathbf{A}^\top\mathbf{Z}^{(k)} =
                     \mathbf{P}^{(k)}\mathbf{W}^{(k),\top} + r^{(k)} e_k^\top\\
\end{align}
where $\mathbf{P}^{(k),\top}\mathbf{P}^{(k)} =
\mathbf{Z}^{(k),\top}\mathbf{Z}^{(k)}=\mathbf{I} \in \R^{k\times k}$, $\mathbf{P}^{(k),\top}r^{(k)}=0$,
$\mathbf{P}^{(k)} e_1=p_1$, and $\mathbf{W}^{(k)}$ is
bidiagonal. Approximate eigenvectors and eigenvalues for $\mathbf A$
can then be computed using the SVD for $\mathbf W$, which is easy
because of the bidiagonal structure. However, the accuracy is closely
tied to the choice of the initial vector $p_1$ and can be measured by
the norm of the residual vector $r^{(k)}$. AIRLB
\citep{BaglamaReichel2005} essentially augments the SVD
of $\mathbf W$ with additional information to reexpress $\mathbf
P^{(k)}$, $\mathbf Z^{(k)}$ and $\mathbf W^{(k)}$. This process is
repeated until the residual is deemed small enough.

If $p_1$ is in the span of the top $j$ eigenvectors of
$\mathbf A$, then no iteration will be necessary, and the
approximation is exact. So, within our ADMM, when the span of the eigenvectors for
$\mathbf B-\mathbf C + \mathbf S/\rho$ is similar across iterations,
previous iterates can be used as initializations for 
the restarted Lanczos procedure. Thus,
as we loop through the ADMM steps, these initializations improve the speed
subsequent decompositions.

This modification significantly improves the
per-iteration efficiency of the algorithm while still converging in a
few dozen iterations.
Note that ADMM will still converge (though in perhaps more iterations)
if some or all of the steps are implemented approximately~\citep{EcksteinBertsekas1992}. In
particular, by approximating the projection of $\mathbf{Q}$, $\mbox{Proj}_{\mathcal{F}^d}(\mathbf{Q})$, with $\widetilde{\mbox{Proj}}_{\mathcal{F}^d}(\mathbf{Q})$, ADMM will converge provided 
\[
\sum_{k=1}^\infty \norm{\mbox{Proj}_{\mathcal{F}^d}(\mathbf{Q}^{(k)}) - \widetilde{\mbox{Proj}}_{\mathcal{F}^d}(\mathbf{Q}^{(k)})} < \infty.
\]
 \citet{NishiharaLessard2015} suggests linear
convergence for ADMM, and our experience is that
this case remains linear despite the approximation.

 \section{Proofs}
\label{sec:proof}

 To prove Theorem 1, we first need the following technical lemma.
\begin{lemma}
  \label{lem:rotate-or-project}  Let $\boldsymbol{\Xi}=\mathbf{ZZ}^\top$ be the orthogonal projector on to the $d$-dimensional
  subspace of $\R^p$ spanned by $\mathbf{Z}$ with $\mathbf{Z}^\top \mathbf{Z}=\mathbf{I}_d$. Let
  $b\in\R^n$ and $\mathbf{A}\in\R^{n\times p}$. Then, defining $\hat x := \argmin_x \norm{\mathbf{AZ} x-b}_2^2$,
  \begin{align}
    \mathbf{Z}\hat x &= \mathbf{Z} (\mathbf{AZ})^+b = (\mathbf{A}\boldsymbol{\Xi})^+b\\
                      &= \argmin_y \norm{\mathbf{A}\boldsymbol{\Xi} y-b}_2^2 =:\hat y,
  \end{align}
  where $\mathbf{Q}^+$ is the Moore-Penrose generalized inverse of $\mathbf{Q}$.
\end{lemma}

\begin{proof}[Proof of Lemma 1]
  We have
  \begin{align}
    \hat y &:= (\mathbf{A}\boldsymbol{\Xi})^+ b =
             \boldsymbol{\Xi}(\mathbf{A}\boldsymbol{\Xi})^+ b =
             \mathbf{ZZ}^\top(\mathbf{AZZ}^\top)^+b\\
           &=\mathbf{ZZ}^\top\mathbf{Z}(
             \mathbf{AZ})^+ b = \mathbf{Z}(\mathbf{AZ})^+b =: \mathbf{Z}\hat x.
  \end{align}
  Here, the first equality is a standard property of idempotent
  matrices and the third follows from \citep[Thm.\ 20.5.6]{Harville1997}.
\end{proof}

\begin{proof}[Proof of Theorem 1]
Write $\boldsymbol{\Pi} = \mathbf{V}_d \mathbf{V}_d^\top$ and $\hat{\boldsymbol{\Pi}} = \hat{\mathbf{V}}_d \hat{\mathbf{V}}_d^\top$ for the orthogonal projectors onto the column span of the population quantity $\mathbf{V}_d$ and the estimate produced by Algorithm 2 respectively. Let
\[
\tilde{\gamma} = \argmin_\gamma \norm{\mathbf{XV}_d\gamma - Y}_2^2,
\]
and define $\tilde{\beta} = \mathbf{V}_d \tilde{\gamma}$. 
Note that $\beta_* \in \mbox{col}(\mathbf{V}_d)$ and $\hat{\beta} \in
\mbox{col}(\hat{\mathbf{V}}_d)$ as
$\beta_*=\mathbf{V}_d\mathbf{L}_d^{-1}\boldsymbol{\Lambda}_d\Theta =:
\mathbf{V_d}r$ and $\hat{\beta}=\mathbf{\hat{V}}_d\hat{\gamma}$.
Then,
\begin{align}
&\norm{\mathbf{X}(\beta_*-\hat{\beta})}_2^2 \\
&= \norm{\mathbf{X}(\beta_*-\tilde{\beta})}_2^2 + \norm{\mathbf{X}(\tilde{\beta}-\hat{\beta})}_2^2\\
&= \norm{\mathbf{X}\left(\mathbf{V}_d \gamma_* - \mathbf{V}_d \tilde{\gamma}\right)}_2^2 + \norm{\mathbf{X}\left(\mathbf{V}_d \tilde{\gamma} - \hat{\mathbf{V}}_d \hat{\gamma}_d\right)}_2^2\\
&= \norm{\mathbf{X}\mathbf{V}_d \left(\gamma_* - \tilde{\gamma}\right)}_2^2 + \norm{\mathbf{X}\left(\boldsymbol{\Pi}\tilde{\beta} - \hat{\boldsymbol{\Pi}} \hat{\beta}\right)}_2^2,
\end{align}
where the last line follows by Lemma 1.

Now, 
$\mathbf{X}\hat\beta = \hat{Y} \in \textrm{col}(X\hat{\boldsymbol{\Pi}})$ 
and $\mathbf{X}\tilde{\beta}=\tilde{Y} \in \textrm{col}(\mathbf{X}\boldsymbol{\Pi})$, and
$\mbox{col}(\mathbf{X}\boldsymbol{\Pi}) \subset \mbox{col}(X) \cap \mbox{row}(\boldsymbol{\Pi})=\mbox{col}(\mathbf{X}) \cap \mbox{col}(\boldsymbol{\Pi})$ (because $\boldsymbol{\Pi}$ is symmetric). Thus, we have that
$\hat{Y} \in \mbox{col}(\mathbf{X}\hat{\boldsymbol{\Pi}}) \subset \mbox{col}(\mathbf{X}) \cap \mbox{col}(\hat{\boldsymbol{\Pi}})$
and $\tilde{Y} \in \mbox{col}(\mathbf{X}\boldsymbol{\Pi}) \subset \mbox{col}(\mathbf{X}) \cap \mbox{col}(\boldsymbol{\Pi})$.
By linearity, there exist orthogonal projectors $\mathbf{Q}$ and $
\mathbf{R}$ such that $\hat{Y} = \mathbf{Q}Y$ and $\tilde{Y} = \mathbf{R}Y$. Therefore,
\begin{align}
\norm{\mathbf{X}\left(\boldsymbol{\Pi}\tilde{\beta} - \hat{\boldsymbol{\Pi}} \hat{\beta}\right)}_2^2 
&= \norm{\hat{Y}-\tilde{Y}}_2^2 \\
&= \norm{(\mathbf{Q}-\mathbf{R})Y}_2^2\\
&\leq \norm{\mathbf{Q}-\mathbf{R}}_F^2 \norm{Y}_2^2 \\
&\leq \norm{\boldsymbol{\Pi}-\hat{\boldsymbol{\Pi}}}_F^2 \norm{Y}_2^2,
\end{align}
where the first inequality is H\"older's inequality for the Frobenius norm. For the second, since $\mathbf{Q}$ and $\mathbf{R}$ are (orthogonal) projectors onto subspaces of $\mbox{col}(\boldsymbol{\Pi})$ and $\mbox{col}(\hat{\boldsymbol{\Pi}})$, it must be that $\norm{\mathbf{Q}-\mathbf{R}}_F \leq \norm{\boldsymbol{\Pi}-\hat{\boldsymbol{\Pi}}}_F$.

Now, since $\hat{\boldsymbol{\Pi}}$ is a solution to the Fantope projection problem under Assumptions A1--A6, we can invoke \citep[Cor. 3.3]{vu2013fantope} to get that 
\[
\norm{\boldsymbol{\Pi}-\hat{\boldsymbol{\Pi}}}_F = \mathcal{O}_P\left(s\sqrt{\log(p)/n}\right)
\]
while $\norm{Y}^2/n=\mathcal{O}_P(1)$. 

Turning now to $\norm{\mathbf{X}\mathbf{V}_d \left(\gamma_* - \tilde{\gamma}\right)}_2^2$, $\tilde\gamma$ is simply the ordinary least squares regression estimate under random design. Thus by, for example, \citep[Theorem 1 and subsequent remarks]{HsuKakade2014}, 
\[
\norm{\mathbf{X}\mathbf{V}_d \left(\gamma_* - \tilde{\gamma}\right)}_2^2 = \mathcal{O}_P\left(\sigma^2 d/n\right).
\]
Combining these terms gives the result.

\end{proof}

\end{document}